\documentclass[letterpaper,11pt]{article}

\usepackage{fullpage}
\usepackage{times}
\usepackage{microtype}
\usepackage{amsmath,amsfonts,amssymb}
\usepackage[amsmath,amsthm,thmmarks,hyperref]{ntheorem}
\usepackage{mathtools}
\usepackage{url}
\usepackage{enumitem}
\usepackage{hyperref}
\usepackage{xspace}
\usepackage{nag}
\usepackage{tikz}
\usetikzlibrary{shapes,arrows,positioning}
\usepackage{comment}

\newcommand{\R}{\ensuremath{\mathbb{R}}}
\newcommand{\T}{\ensuremath{\mathbb{T}}}
\newcommand{\Z}{\ensuremath{\mathbb{Z}}}

\DeclarePairedDelimiter\inner{\langle}{\rangle}
\DeclarePairedDelimiter\abs{\lvert}{\rvert}
\DeclarePairedDelimiter\set{\{}{\}}
\DeclarePairedDelimiter\parens{(}{)}
\DeclarePairedDelimiter\bracks{[}{]}
\DeclarePairedDelimiter\floor{\lfloor}{\rfloor}
\DeclarePairedDelimiter\length{\lVert}{\rVert}
\newcommand{\matA}{\ensuremath{\mathbf{A}}}
\newcommand{\matB}{\ensuremath{\mathbf{B}}}
\newcommand{\matC}{\ensuremath{\mathbf{C}}}

\newcommand{\matG}{\ensuremath{\mathbf{G}}}

\newcommand{\matI}{\ensuremath{\mathbf{I}}}

\newcommand{\matP}{\ensuremath{\mathbf{P}}}

\newcommand{\matR}{\ensuremath{\mathbf{R}}}

\newcommand{\matU}{\ensuremath{\mathbf{U}}}

\newcommand{\veca}{\ensuremath{\mathbf{a}}}
\newcommand{\vecb}{\ensuremath{\mathbf{b}}}
\newcommand{\vecc}{\ensuremath{\mathbf{c}}}

\newcommand{\vece}{\ensuremath{\mathbf{e}}}
\newcommand{\vecf}{\ensuremath{\mathbf{f}}}

\newcommand{\vecp}{\ensuremath{\mathbf{p}}}

\newcommand{\vecr}{\ensuremath{\mathbf{r}}}
\newcommand{\vecs}{\ensuremath{\mathbf{s}}}

\newcommand{\vecu}{\ensuremath{\mathbf{u}}}
\newcommand{\vecv}{\ensuremath{\mathbf{v}}}

\newcommand{\vecx}{\ensuremath{\mathbf{x}}}
\newcommand{\vecy}{\ensuremath{\mathbf{y}}}
\newcommand{\vecz}{\ensuremath{\mathbf{z}}}
\newcommand{\veczero}{\ensuremath{\mathbf{0}}}

\theoremstyle{plain}            % following are "theorem" style
\newtheorem{theorem}{Theorem}[section]
\newtheorem{lemma}[theorem]{Lemma}
\newtheorem{corollary}[theorem]{Corollary}

\newtheorem{claim}[theorem]{Claim}

\theoremstyle{definition}       % following are def style
\newtheorem{definition}[theorem]{Definition}

\theoremstyle{remark}           % following are remark style

\numberwithin{equation}{section}

\DeclareMathOperator{\poly}{poly}

\newcommand{\algo}[1]{\ensuremath{\mathsf{#1}}\xspace}

\newcommand{\problem}[1]{\ensuremath{\mathsf{#1}}\xspace}

\newcommand{\lat}{\mathcal{L}}
\newcommand{\smooth}{\eta}
\newcommand{\smootheps}{\smooth_{\epsilon}}
\newcommand{\gs}[1]{\ensuremath{\widetilde{#1}}}
\newcommand{\sampleD}{\algo{SampleD}}

\newcommand{\gapsvp}{\problem{GapSVP}}

\newcommand{\usvp}{\problem{uSVP}}
\newcommand{\sivp}{\problem{SIVP}}

\newcommand{\bdd}{\problem{BDD}}

\newcommand{\sis}{\problem{SIS}}

\newcommand{\lwe}{\problem{LWE}}
\newcommand{\rlwe}{\problem{RLWE}}

\newcommand{\extlwe}{\problem{extLWE}}

\newcommand{\calZ}{\mathcal{Z}}

\usepackage{amsbsy}
\usepackage[mathscr]{eucal}

\newif\ifnotes\notesfalse

\ifnotes
\usepackage{color}
\definecolor{mygrey}{gray}{0.50}
\newcommand{\notename}[2]{{\textcolor{mygrey}{\footnotesize{\bf (#1:} {#2}{\bf ) }}}}
\newcommand{\noteswarning}{{\begin{center} {\Large WARNING: NOTES ON}\end{center}}}

\else

\newcommand{\notename}[2]{{}}
\newcommand{\noteswarning}{{}}

\fi

\newcommand{\cnote}[1]{{\notename{Chris}{#1}}}
\newcommand{\znote}[1]{{\notename{Zvika}{#1}}}
\newcommand{\onote}[1]{{\notename{Oded}{#1}}}

\newcommand{\dnote}[1]{{\notename{Damien}{#1}}}

\newcommand{\eps}{\varepsilon}

\renewcommand{\epsilon}{\varepsilon}

\newcommand{\myparagraph}[1]{\paragraph{{#1}.}}

\def\binset{\{0,1\}}
\def\bbZ{{\mathbb Z}}
\def\bbN{{\mathbb N}}
\def\bbT{{\mathbb T}}
\def\cA{{\cal A}}
\def\cB{{\cal B}}
\def\cZ{{\cal Z}}

\def\getsr{\gets}
\def\getsd{{:=}}

\newcommand{\mx}[1]{\mathbf{{#1}}}
\newcommand{\vc}[1]{\mathbf{{#1}}}

\newcommand{\adv}{\mathrm{Adv}}

\newcommand{\zolwe}{\mathsf{binLWE}}
\def\mAh{\hat{\mx{A}}}

\renewcommand{\vec}{\mathbf}
\newcommand{\calD}{\mathcal{D}}

\title{\bf Classical Hardness of Learning with Errors}

\date{}
\author{
Zvika Brakerski\thanks{Stanford University, \texttt{zvika@stanford.edu}. Supported by a Simons Postdoctoral Fellowship and DARPA.}
\and
Adeline Langlois
\thanks{ENS de Lyon and Laboratoire LIP (U.\ Lyon, CNRS, ENS Lyon, INRIA, UCBL),
    46 All\'ee d'Italie, 69364 Lyon Cedex 07, France. \texttt{adeline.langlois@ens-lyon.fr}.}
\and
Chris Peikert\thanks{School of Computer Science, College of
    Computing, Georgia Institute of Technology.
    This material is based upon work supported by the National Science
    Foundation under CAREER Award~CCF-1054495, by DARPA under
    agreement number FA8750-11-C-0096, and by the Alfred P.~Sloan
    Foundation.  Any opinions, findings, and conclusions or
    recommendations expressed in this material are those of the
    author(s) and do not necessarily reflect the views of the National
    Science Foundation, DARPA or the U.S.~Government, or the Sloan
    Foundation.  The U.S. Government is authorized to reproduce and
    distribute reprints for Governmental purposes notwithstanding any
    copyright notation thereon.}
  \and Oded Regev\thanks{Courant Institute, New York
    University. Supported by a European Research Council (ERC)
    Starting Grant. Part of the work done while the author was with
    the CNRS, DI, ENS, Paris.}
\and
Damien Stehl\'e \thanks{ENS de Lyon and Laboratoire LIP (U.\ Lyon,
  CNRS, ENS Lyon, INRIA, UCBL), 46 All\'ee d'Italie, 69364 Lyon Cedex
  07, France. \texttt{damien.stehle@ens-lyon.fr}.  The author was
  partly supported by the Australian Research Council Discovery Grant
  DP110100628.  }
}

\begin{document}

\maketitle

\begin{abstract}
  We show that the Learning with Errors (LWE) problem is \emph{classically} at least as hard as standard worst-case lattice problems, even with polynomial modulus. 
Previously this was only known under \emph{quantum} reductions.

Our techniques capture the tradeoff between the dimension and the modulus of LWE instances, leading to a much better understanding of the
landscape of the problem. The proof is inspired by techniques from
several recent cryptographic constructions, most notably fully
homomorphic encryption schemes.

%%% Local Variables:
%%% mode: latex
%%% TeX-master: "lwehard"
%%% End:

\end{abstract}

%%%% DON'T REMOVE %%%%%
\noteswarning
%%%% DON'T REMOVE %%%%%

\section{Introduction}
\label{sec:introduction}

Over the last decade, lattices have emerged as a very attractive
foundation for cryptography.  The appeal of lattice-based primitives
stems from the fact that their security can be based on
\emph{worst-case} hardness assumptions, that they appear to remain
secure even against \emph{quantum} computers, that they can be
quite efficient, and that, somewhat surprisingly, for certain advanced
tasks such as fully homomorphic encryption no other cryptographic
assumption is known to suffice.

Virtually all recent lattice-based cryptographic schemes are based
directly upon one of two natural average-case problems that have
been shown to enjoy worst-case hardness guarantees: the
\emph{short integer solution} ($\sis$) problem and the
\emph{learning with errors} ($\lwe$) problem.
The former dates back to Ajtai's groundbreaking
work~\cite{Ajtai96}, who showed that it is at least as
hard as approximating several worst-case lattice problems, such as the
(decision version of the) shortest vector problem, known as $\gapsvp$, to within a
polynomial factor in the lattice dimension. This hardness result
was tightened in followup work (e.g.,~\cite{DBLP:journals/siamcomp/MicciancioR07}),
leading to a somewhat satisfactory understanding of the hardness of the $\sis$ problem.
The $\sis$ problem
has been the foundation for
one-way~\cite{Ajtai96} and
collision-resistant hash
functions~\cite{DBLP:journals/eccc/ECCC-TR96-042}, identification
schemes~\cite{DBLP:conf/crypto/MicciancioV03,DBLP:conf/pkc/Lyubashevsky08,DBLP:conf/asiacrypt/KawachiTX08},
and digital
signatures~\cite{DBLP:conf/stoc/GentryPV08,DBLP:conf/eurocrypt/CashHKP10,DBLP:conf/pkc/Boyen10,DBLP:conf/eurocrypt/MicciancioP12,DBLP:conf/eurocrypt/Lyubashevsky12}. 

Our focus in this paper is on the latter problem, learning with errors.
In this problem our goal is to distinguish with some non-negligible advantage between the following
two distributions:
\[
\parens{\parens{ \veca_i, \inner{\vc{a}_i, \vc{s}}+e_i \bmod q}}_{i}
\quad \text{and} \quad
\parens{\parens{ \veca_i, u_i }}_{i}~,
\]
where $\vecs$ is chosen uniformly from $\Z_q^n$
and so are the $\veca_i \in \Z_q^n$,  $u_i$ are chosen uniformly from $\Z_q$,
and the ``noise'' $e_i \in \Z$ is sampled from some distribution supported on small numbers, typically a (discrete) Gaussian
distribution with standard deviation $\alpha q$ for $\alpha = o(1)$.
%\dnote{I replaced $\ll$ by $o(\cdot)$}

The $\lwe$ problem has proved to be amazingly versatile, serving as
the basis for a multitude of cryptographic constructions:
secure public-key encryption under both
chosen-plaintext~\cite{DBLP:journals/jacm/Regev09,DBLP:conf/crypto/PeikertVW08,DBLP:conf/ctrsa/LindnerP11}
and
chosen-ciphertext~\cite{DBLP:conf/stoc/PeikertW08,DBLP:conf/stoc/Peikert09,DBLP:conf/eurocrypt/MicciancioP12}
attacks, oblivious transfer~\cite{DBLP:conf/crypto/PeikertVW08},
identity-based
encryption~\cite{DBLP:conf/stoc/GentryPV08,DBLP:conf/eurocrypt/CashHKP10,DBLP:conf/eurocrypt/AgrawalBB10,DBLP:conf/crypto/AgrawalBB10},
various forms of leakage-resilient cryptography (e.g.,
\cite{DBLP:conf/tcc/AkaviaGV09,DBLP:conf/crypto/ApplebaumCPS09,DBLP:conf/innovations/GoldwasserKPV10}),
fully homomorphic
encryption~\cite{DBLP:conf/focs/BrakerskiV11,DBLP:conf/innovations/BrakerskiGV12,DBLP:conf/crypto/Brakerski12}
(following the seminal work of Gentry~\cite{DBLP:conf/stoc/Gentry09}),
and much more. It was also used to show hardness of learning problems~\cite{DBLP:conf/focs/KlivansS06}.

Contrary to the $\sis$ problem, however, the hardness of $\lwe$ is not
sufficiently well understood. The main hardness reduction for
$\lwe$~\cite{DBLP:journals/jacm/Regev09} is similar to the one for $\sis$
mentioned above, except that it is \emph{quantum}. This means that
the existence of an efficient algorithm for $\lwe$, even
a classical (i.e., non-quantum) one, only implies the existence
of an efficient \emph{quantum} algorithm for lattice problems.
This state of affairs is quite unsatisfactory: even though
one might conjecture that efficient quantum algorithms for lattice
problems do not exist, our understanding of quantum algorithms
is still at its infancy. It is therefore highly desirable to
come up with a \emph{classical} hardness reduction for $\lwe$.

Progress in this direction was made
by~\cite{DBLP:conf/stoc/Peikert09} (with some simplifications in the
followup by Lyubashevsky and
Micciancio~\cite{DBLP:conf/crypto/LyubashevskyM09}).  The main result
there is that $\lwe$ with \emph{exponential} modulus is as hard as
some standard lattice problems using a
classical reduction.  
%\dnote{I explained which problems in the previous sentence} 
As that hardness result crucially relies on the
exponential modulus, the open question remained as to whether $\lwe$
is hard for smaller moduli, in particular polynomial moduli. In
addition to being an interesting question in its own right, this question is of special importance since many cryptographic applications, as well as the learning theory
result of Klivans and Sherstov~\cite{DBLP:conf/focs/KlivansS06}, are instantiated in this setting.
%In
%addition to being an interesting question in its own right, we note
%that most applications use a polynomial modulus; although many
%applications can be extended to deal with an exponential modulus, a
%polynomial modulus is sometimes necessary, say in the learning theory
%result of Klivans and Sherstov~\cite{DBLP:conf/focs/KlivansS06}.
%\onote{add the ut ot?}
Some additional evidence that reducing the
modulus is a fundamental question comes from the Learning Parity with
Noise (LPN) problem, which can be seen as $\lwe$ with modulus $2$
(albeit with a different error distribution), and whose hardness is a
long-standing open question.  We remark
that~\cite{DBLP:conf/stoc/Peikert09} does include a classical hardness
of $\lwe$ with polynomial modulus, albeit one based on a
non-standard lattice problem, whose hardness is arguably as debatable
as that of the $\lwe$ problem itself.

To summarize, prior to our work, the existence of
an efficient algorithm for $\lwe$ with polynomial modulus was only known to imply an efficient \emph{quantum}
algorithm for lattice problems, or an efficient classical algorithm for a
non-standard lattice problem. While both consequences are unlikely,
they are arguably not as earth-shattering as an efficient classical
algorithm for lattice problems. Hence, some concern about the hardness of
$\lwe$ persisted, tainting the plethora of cryptographic applications
based on it.

\myparagraph{Main result}
We provide the first classical hardness reduction of $\lwe$ with polynomial modulus. Our reduction is
the first to show that the existence of an efficient classical algorithm for $\lwe$
with any subexponential modulus would indeed have earth-shattering consequences: it would imply an efficient algorithm
for worst-case instances of standard lattice problems.

\onote{should we add somewhere a formal statement of our main result? at least a combination of sections 3 and 4?}
\begin{theorem}[Informal]\label{thm:informal}
Solving $n$-dimensional $\lwe$ with $\poly(n)$ modulus
implies an equally efficient solution to a worst-case lattice problem in dimension $\sqrt{n}$.
\end{theorem}
As a result, we establish the hardness of all known applications of polynomial-modulus $\lwe$ based on
classical worst-case lattice problems, previously only known under a quantum assumption.

\myparagraph{Techniques}
Even though our main theorem has the flavor of a statement in computational complexity,
its proof crucially relies on a host of ideas coming from recent progress in cryptography,
most notably recent breakthroughs in the construction of fully homomorphic encryption schemes.

At a high level, our main theorem is a ``modulus reduction'' result:
we show a reduction from $\lwe$ with large modulus $q$ and dimension $n$ to
$\lwe$ with (small) modulus $p=\poly(n)$ and dimension $n \log_2 q$. Theorem~\ref{thm:informal}
now follows from the main result in~\cite{DBLP:conf/stoc/Peikert09}, which shows that
the former problem with $q=2^n$ is as hard as $n$-dimensional $\gapsvp$.
We note that the increase in dimension from $n$ to $n \log_2 q$ is to be expected,
as it essentially preserves the number of possible secrets (and hence the running time
of the naive brute-force algorithm).%invariant

Very roughly speaking,
the main idea in modulus reduction is to map $\Z_q$ into $\Z_p$ through the naive mapping
that sends any $a \in \{0,\ldots,q-1\}$ to $\floor{pa/q} \in \{0,\ldots,p-1\}$.
This basic idea is confounded by two issues. The first is that if carried out naively, this
transformation introduces rounding artifacts into $\lwe$, ruining the distribution
of the output. We resolve this issue by using a more careful Gaussian randomized rounding procedure (Section~\ref{sec:modexpansion}).
A second serious issue is that in order for the rounding errors not to be
amplified when multiplied by the $\lwe$ secret $\vecs$, it is essential to assume that $\vecs$
has small coordinates. A major part of our reduction (Section~\ref{sec:shortsecret}) is therefore dedicated to showing
a reduction from $\lwe$ (in dimension $n$) with arbitrary secret in $\Z_q^n$ to $\lwe$ (in dimension $n \log_2 q$)
with a secret chosen uniformly over $\{0,1\}$. This follows from a careful hybrid argument (Section~\ref{sec:theshortsecretreduction}) combined with a
hardness reduction to the so-called ``extended-$\lwe$'' problem, which is a variant of $\lwe$ in which we have some control
over the error vector (Section~\ref{sec:extlwe}).

We stress that even though our proof is inspired by and has analogues
in the cryptographic literature, the details of the reductions are very different.
In particular, the idea of modulus reduction plays a key role in recent work on fully homomorphic
encryption schemes, giving a way to control the noise growth during homomorphic
operations~\cite{DBLP:conf/focs/BrakerskiV11,DBLP:conf/innovations/BrakerskiGV12,DBLP:conf/crypto/Brakerski12}.
However, since the goal there is merely to preserve the functionality of the scheme,
their modulus reduction can be performed in a rather naive way similar to the one outlined above,
and so the output of their procedure does not constitute a valid $\lwe$ instance.
In our reduction we need to perform a much more delicate modulus reduction,
which we do using Gaussian randomized rounding, as mentioned above.

The idea of reducing $\lwe$ to have a $\{0,1\}$ secret also exists already in the cryptographic literature:
precisely such a reduction was shown by 
Goldwasser et al.~\cite{DBLP:conf/innovations/GoldwasserKPV10}
who were motivated by questions in leakage-resilient cryptography.
Their reduction, however, incurred a severe blow-up in the noise rate, making it useless for our purposes. 
In more detail, not being able to
faithfully reproduce the $\lwe$ distribution in the output, they resort to 
hiding the faults in the output distribution under a huge independent 
fresh noise, in order to make it close to the correct one. The trouble 
with this ``noise flooding'' approach
%\dnote{I reformulated these sentences to make ``noise flooding'' less 
% mysterious}
%what has been called ``noise flooding'' in order to hide the faults
%in the output distribution, and to make it close to the correct one. \znote{Maybe say what noise flooding is? It seems somewhat mysterious like this.}
%The trouble with this approach 
is that the amount of noise
one has to add depends on the running time of the algorithm solving
the target $\{0,1\}$-$\lwe$ problem, which in turn forces the modulus
to be equally big.
So while in principle we could use the reduction
from~\cite{DBLP:conf/innovations/GoldwasserKPV10} (and shorten our proof
by about a half), this would lead to a qualitatively much weaker result:
the modulus and the approximation ratio for the worst-case lattice problem
would both grow with the running time of the $\{0,1\}$-$\lwe$ algorithm.
In particular, we would not be able to show that for some fixed
polynomial modulus, $\lwe$ is a hard problem; instead, in order to capture all polynomial time algorithms,
we would have to take a super-polynomial
modulus, and rely on the hardness of worst-case lattice problem to within super-polynomial
approximation factors.
In contrast, with our reduction,
the modulus and the approximation ratio both remain fixed independently of the target $\{0,1\}$-$\lwe$ algorithm.

As mentioned above, our alternative to the reduction
in~\cite{DBLP:conf/innovations/GoldwasserKPV10} is based
on a hybrid argument combined with a new hardness reduction
for the ``extended LWE'' problem, which is a variant of $\lwe$ in which in addition to the LWE samples, we also get to see the inner product
of the vector of error terms with a vector $\vecz$ of our 
%\dnote{Shall we add (``adversorial'')?}\onote{i prefer to keep as is} 
choosing. This problem has
its origins in the cryptographic literature, namely
in the work of O'Neill, Peikert, and Waters~\cite{DBLP:conf/crypto/ONeillPW11} on (bi)deniable encryption
and the later work of Alperin-Sheriff and Peikert~\cite{DBLP:conf/pkc/Alperin-SheriffP12}
on key-dependent message security. The hardness reductions included in those
papers are not sufficient for our purposes, as they cannot handle
large moduli or error terms, which is crucial in our setting. We therefore
provide an alternative reduction which is conceptually much simpler, and
essentially subsumes both previous reductions.
Our reduction works equally well with
exponential moduli and correspondingly long error vectors, a case earlier reductions could not handle.

\myparagraph{Broader perspective}
As a byproduct of the proof of Theorem~\ref{thm:informal}, we obtain several results that
shed new light on the hardness of $\lwe$. Most notably, our modulus reduction result
in Section~\ref{sec:modexpansion} is actually far more general, and can be used to show
a ``modulus expansion/dimension reduction'' tradeoff. Namely, it shows a reduction
from $\lwe$ in dimension $n$ and modulus $p$ to $\lwe$ in dimension $n/k$ and modulus $p^k$ (see Corollary~\ref{cor:mod-dim-tradeoff}).
Combined with our modulus reduction, this has the following interesting consequence:
the hardness of $n$-dimensional $\lwe$ with modulus $q$ is a function of the quantity $n \log_2 q$.
In other words, varying $n$ and $q$ individually while keeping $n \log_2 q$ fixed
essentially preserves the hardness of $\lwe$.

Although we find this statement quite
natural (since $n \log_2 q$ represents the number of bits in the secret), it has
some surprising consequences. One is that $n$-dimensional $\lwe$ with modulus $2^n$
is essentially as hard as $n^2$-dimensional $\lwe$ with polynomial modulus.
As a result, $n$-dimensional $\lwe$ with modulus $2^n$, which was shown
in~\cite{DBLP:conf/stoc/Peikert09} to be as hard as $n$-dimensional lattice
problems using a classical reduction, is actually as hard as $n^2$-dimensional
lattice problems using a quantum reduction. The latter is presumably a much
harder problem, requiring $\exp(\widetilde\Omega(n^2))$ time to solve.
This corollary highlights an inherent quadratic loss in the classical reduction
of~\cite{DBLP:conf/stoc/Peikert09} (and as a result also our Theorem~\ref{thm:informal})
compared to the quantum one in~\cite{DBLP:journals/jacm/Regev09}.

\onote{for the final version, consider adding precise statements for this and the next point}
A second interesting consequence is that $1$-dimensional $\lwe$ with
modulus $2^n$ is essentially as hard as $n$-dimensional $\lwe$ with
polynomial modulus. The $1$-dimensional version of $\lwe$ is closely
related to the Hidden Number Problem of Boneh and
Venkatesan~\cite{DBLP:conf/crypto/BonehV96}. It is also essentially equivalent to the
Ajtai-Dwork-type~\cite{DBLP:conf/stoc/AjtaiD97} cryptosystem in~\cite{DBLP:journals/jacm/Regev04}, as follows from simple reductions
similar to the one in the appendix of~\cite{RegevSurvey}.
Moreover, the $1$-dimensional version can be seen as a special case of
the Ring-$\lwe$ problem introduced
in~\cite{DBLP:conf/eurocrypt/LyubashevskyPR10} (for ring dimension~1,
i.e., ring equal to~$\Z$). This allows us, via the ring switching
technique from~\cite{DBLP:conf/scn/GentryHPS12}, to obtain the first
hardness proof of Ring-$\lwe$, with arbitrary ring dimension and
exponential modulus, under the hardness of problems on general lattices (as
opposed to just ideal lattice problems). In addition, this leads to the
first hardness proof for the Ring-$\sis$
problem~\cite{DBLP:conf/icalp/LyubashevskyM06,DBLP:conf/tcc/PeikertR06}
with exponential modulus under the hardness of general lattice problems, via the standard
$\lwe$-to-$\sis$ reduction. (We note that since both results are obtained by
scaling up from a ring of dimension $1$, the hardness does not improve
as the ring dimension increases.)

A final interesting consequence of our reductions is that (the
decision form of) \lwe is hard with an arbitrary huge modulus, e.g., a
prime; see Corollary~\ref{cor:mod-reduction-2}.  Previous results
(e.g.,~\cite{DBLP:journals/jacm/Regev09,DBLP:conf/stoc/Peikert09,DBLP:conf/crypto/MicciancioM11,DBLP:conf/eurocrypt/MicciancioP12})
required the modulus to be \emph{smooth}, i.e., all its prime divisors
had to be polynomially bounded.

\onote{for final version: decide what to do with real-LWE}

\onote{for final version, consider applying these ideas to SIS, especially the 1-dim case which is like subset sum, and more generally, modulus/dimension tradeoff for SIS. Chris said:
The idea is to do a similar transformation as in Lemma 4.4, but on SIS samples, going from a in $\Z_q^n$ to $a' \in \Z_{q'}^{n'}$.  Then given a good SIS solution $x$ to $A'x = 0$, we also have $G^T A' x = 0$ and therefore $A x \approx 0$.  Thus $x$ is a solution to the original SIS instance, because it suffices to find a preimage of "nearly" 0.  This is because if $Ax = e \approx 0$, then we get the homogeneous solution $[I | A] (e; x) = 0$, and the instance $[I | A]$ is just SIS in (hermite) normal form.}

\dnote{Something that I find interesting is that the complexity of the
  best known attack is~$2^{O(n \log q/ \log^2 \alpha)}$, which is
  constant wrt $n \log q$. This was not reflected by our previous
  understanding of the complexity of \lwe, and we manage to close the
  gap between the hardness results and the algorithms. By the way,
  this~$2^{O(n \log q/ \log^2 \alpha)}$ bound suggests it may be
  possible to transfer some of~$n$ and~$\log q$ into~$\log \alpha$. Do
  we know \lwe self-reductions from, say, $q,\alpha$ to~$p,\beta$ such
  that~$\log q/ |\log \alpha| = \log p/|\log \beta|$? or from
  $n,\alpha$ to~$n',\beta$ such that~$n/ |\log \alpha| = n'/|\log
  \beta|$}

\onote{maybe mention Ajtai's 2005 paper using Dirichlet's problem?}

\myparagraph{Open questions}
As mentioned above, our Theorem~\ref{thm:informal} inherits from~\cite{DBLP:conf/stoc/Peikert09}
a quadratic loss in the dimension, which does not exist in the quantum reduction~\cite{DBLP:journals/jacm/Regev09}
nor in the known hardness reductions for $\sis$. 
At a technical level, this quadratic loss stems from the fact that the reduction in~\cite{DBLP:conf/stoc/Peikert09} is not iterative. In contrast,
the quantum reduction in~\cite{DBLP:journals/jacm/Regev09} as well as the reductions for $\sis$ are iterative, and as a result do not incur
the quadratic loss. We note that an additional side effect of the non-iterative reduction is that the hardness in Theorem~\ref{thm:informal} 
and~\cite{DBLP:conf/stoc/Peikert09} is based only on the worst-case lattice problem \gapsvp (and the essentially equivalent \bdd and \usvp~\cite{DBLP:conf/crypto/LyubashevskyM09}), and not on problems like \sivp, which the quantum reduction of~\cite{DBLP:journals/jacm/Regev09} and the hardness
reductions for \sis can handle. One case where this is very significant is when dealing with ideal lattices, as in the hardness 
reduction for Ring-\lwe, since \gapsvp turns out to be an easy problem there. 

We therefore believe that it is important to understand whether there exists a classical
reduction that does not incur the quadratic loss inherent in~\cite{DBLP:conf/stoc/Peikert09}
and in Theorem~\ref{thm:informal}.
In other words, is $n$-dimensional $\lwe$ with polynomial modulus classically as hard as $n$-dimensional
lattice problems (as opposed to $\sqrt{n}$-dimensional)? This would constitute
the first full dequantization of the quantum reduction in~\cite{DBLP:journals/jacm/Regev09}.

While it is natural to conjecture that the answer to this question is positive,
a negative answer would be quite tantalizing. In particular, it is
conceivable that there exists a (classical) algorithm for $\lwe$ with polynomial modulus running in time
$2^{O(\sqrt{n})}$. Due to the quadratic expansion in Theorem~\ref{thm:informal}, this
would not lead to a faster classical algorithm for lattice problems; it would,
however, lead to a $2^{O(\sqrt{n})}$-time \emph{quantum} algorithm for
lattice problems using the reduction in~\cite{DBLP:journals/jacm/Regev09}.
The latter would be a major progress in quantum algorithms, yet is not
entirely unreasonable; in fact, a $2^{O(\sqrt{n})}$-time quantum algorithm
for a somewhat related quantum task was discovered by
Kuperberg~\cite{DBLP:journals/siamcomp/Kuperberg05} (see also~\cite{DBLP:journals/siamcomp/Regev04}).

%%% Local Variables:
%%% mode: latex
%%% TeX-master: "lwehard"
%%% End:

\section{Preliminaries}
\label{sec:preliminaries}

Let $\T =
\R/\Z$ denote the cycle, i.e., the additive group of reals modulo
$1$. We also denote by $\T_q$ its cyclic subgroup of order $q$, i.e.,
the subgroup given by $\{0,1/q,\ldots,(q-1)/q\}$.

For two probability distributions $P,Q$ over some discrete domain, we
define their statistical distance as $\sum |P(i)-Q(i)|/2$ where~$i$
ranges over the distribution domain, and extend this to continuous
distributions in the obvious way. We recall the following easy fact
(see, e.g.,~\cite[Eq.~(2.3)]{AldousDiaconis87} for a proof).
\begin{claim}\label{clm:separationdist}
  If $P$ and $Q$ are two probability distributions such that~$P(i) \ge
  (1-\eps)Q(i)$ holds for all~$i$, then the statistical distance
  between $P$ and $Q$ is at most~$\eps$.
\end{claim}

We will use the following immediate corollary of the leftover hash
lemma~\cite{DBLP:journals/siamcomp/HastadILL99}.

\begin{lemma}
  \label{lem:lhl}
  Let $k, n, q \ge 1$ be integers, and $\epsilon > 0$ be such that $n
  \ge k \log_2 q + 2 \log_2 (1/\epsilon)$. For $\mx{H} \gets \T_q^{k \times
    n}$, $\vc{z} \gets \binset^n$, $\vc{u} \gets \bbT_q^k$, the
  distributions of $(\mx{H}, \mx{H}\vc{z})$ and $(\mx{H}, \vc{u})$ are
  within statistical distance at most $\epsilon$.
\end{lemma}

A \emph{distinguishing problem} $P$ is defined by two distributions $P_0$ and $P_1$, and a solution to the problem is the ability to distinguish between these distributions. The \emph{advantage} of an algorithm $\cA$ with binary output on $P$ is defined as
\[
\adv[\cA] = \abs{\Pr[\cA(P_0)] - \Pr[\cA(P_1)]}~.
\]
A reduction from a problem~$P$ to a problem~$Q$ is an efficient (i.e.,
polynomial-time) algorithm~$\cA^{\cB}$ that solves~$P$ given
access to an oracle~$\cB$ that solves~$Q$. Most of our reductions
(in fact all except the one in Lemma~\ref{lem:lelwe}) are what we 
call ``transformation reductions:'' these reductions perform some transformation
to the input and then apply the oracle to the result.

\subsection{Lattices}
\label{sec:lattices}

An $n$-dimensional (full-rank) lattice $\Lambda \subseteq \R^n$ is the
set of all integer linear combinations of some set of $n$ linearly
independent \emph{basis} vectors $\matB = \set{\vecb_{1}, \ldots,
  \vecb_{n}} \subseteq \R^n$,
\[ \Lambda = \lat(\matB) = \set[\Big]{
  \sum_{i \in [n]} z_i \vecb_i\; :\; \vecz \in \mathbb{Z}^n }. \]
The \emph{dual lattice} of $\Lambda \subset \R^n$ is defined as
$\Lambda^* = \set{ \vecx \in \R^n : \inner{\Lambda,\vecx} \subseteq \Z}$.

The \emph{minimum distance} (or \emph{first successive minimum}) $\lambda_1(\Lambda)$ of a lattice
$\Lambda$ is the length of a shortest nonzero lattice vector, i.e., $\lambda_1(\Lambda) = \min_{\veczero \neq
  \vecx \in \Lambda} \length{\vecx}$.  For an approximation ratio $\gamma = \gamma(n) \ge 1$, the $\gapsvp_\gamma$ is the problem of deciding,
	given a basis $\matB$ of an $n$-dimensional lattice $\Lambda = \lat(\matB)$ and a number $d$, between the case where $\lambda_1(\lat(\matB)) \le d$ 
	and the case where $\lambda_1(\lat(\matB)) > \gamma d$.
We refer to~\cite{KhotLLL25,RegevLLL25} for a recent account 
on the computational complexity of $\gapsvp_\gamma$. \dnote{I added that sentence}

\subsection{Gaussian measures}
\label{sec:gaussian_measures}

For $r > 0$, the $n$-dimensional Gaussian function $\rho_r : \R^{n} \to (0,1]$ is
defined as
\[ \rho_r(\vecx) := \exp(-\pi
\length{\vecx}^{2}/r^2).
\]
We extend this definition to sets, i.e., $\rho_r(A) = \sum_{\vecx \in A} \rho_r(\vecx) 
\in [0,+\infty]$ for any~$A \subseteq \R^{n}$.
The (spherical) continuous Gaussian distribution~$D_r$ is the distribution with density
function proportional to $\rho_r$. More generally, for a matrix $\matB$,
we denote by $D_\matB$ the distribution of $\matB \vecx$ where $\vecx$ is sampled
from $D_1$. When $\matB$ is nonsingular, its probability density function is
proportional to
\[
\exp(-\pi \vecx^T (\matB \matB^T)^{-1} \vecx).
\]
A basic fact is that for any matrices $\matB_1,\matB_2$, the sum of a sample from $D_{\matB_1}$ and an independent sample
from~$D_{\matB_2}$ is distributed like $D_{\matC}$ for $\matC = (\matB_1^{} \matB_1^T + \matB_2^{} \matB_2^T)^{1/2}$.

For an $n$-dimensional lattice $\Lambda$ and a vector $\vecu \in \R^n$, we define the \emph{discrete Gaussian distribution}~$D_{\Lambda+\vecu,r}$
as the discrete distribution with support on the coset $\Lambda+\vecu$ whose probability mass function is proportional to
$\rho_r$. There exists an efficient procedure that samples within negligible statistical distance of any (not too narrow) discrete Gaussian distribution 
(\cite[Theorem~4.1]{DBLP:conf/stoc/GentryPV08}; see also~\cite{DBLP:conf/crypto/Peikert10}). In the next lemma, proved in Section~\ref{sec:exactgpv},
we modify this sampler so that the output is distributed exactly as a discrete Gaussian. This also allows us to sample from
slightly narrower Gaussians. Strictly speaking, the lemma is not needed for our results,
and we could use instead the original sampler from \cite{DBLP:conf/stoc/GentryPV08}. Using our exact sampler leads to slightly
cleaner proofs as well as a (miniscule) improvement in the parameters of our reductions, and we include it here mainly in the hope that it finds
further applications in the future.

\begin{lemma}\label{lem:gpv}
There is a probabilistic polynomial-time algorithm that, given a basis $\matB$ of an $n$-dimensional
lattice $\Lambda = \lat(\matB)$, $\vecc \in \R^n$, and a parameter $r \ge \|\widetilde \matB\| \cdot \sqrt{\ln(2n+4)/\pi}$,
outputs a sample distributed according to $D_{\Lambda+\vecc,r}$. 
\end{lemma}

Here, $\widetilde \matB$ denotes the Gram-Schmidt orthogonalization of $\matB$,
and $\|\widetilde \matB\|$ is the length of the longest vector in it.
We recall the definition of the \emph{smoothing parameter} from~\cite{DBLP:journals/siamcomp/MicciancioR07}.

\begin{definition}
  \label{def:smoothing}
  For a lattice $\Lambda$ and positive real $\epsilon > 0$, the
  smoothing parameter $\smootheps(\Lambda)$ is the smallest real $s >
  0$ such that $\rho_{1/s}(\Lambda^* \setminus \set{\veczero}) \leq
  \epsilon$.
\end{definition}

\begin{lemma}[{{\cite[Lemma~3.1]{DBLP:conf/stoc/GentryPV08}}}]\label{lem:boundonsmoothing}
For any $\eps>0$ and $n$-dimensional lattice $\Lambda$ with basis $\matB$,
\[
\smootheps(\Lambda) \le \|\widetilde \matB \| \sqrt{\ln(2n(1+1/\eps))/\pi}.
\]
\end{lemma}

We now collect some known facts on Gaussian distributions and lattices.

\begin{lemma}[{{\cite[Lemma 4.1]{DBLP:journals/siamcomp/MicciancioR07}}}]
  \label{lem:smoothingcontinuous}
  For any $n$-dimensional lattice $\Lambda$, $\epsilon>0$, $r \geq
  \eta_{\epsilon}(\Lambda)$, the distribution of $\vecx \bmod \Lambda$ where $\vecx \leftarrow D_{r}$ is
	within statistical distance $\eps/2$ of the uniform distribution on cosets of $\Lambda$.
\end{lemma}

\begin{lemma}[{{\cite[Claim 3.8]{DBLP:journals/jacm/Regev09}}}]
  \label{lem:smoothing}
  For any $n$-dimensional lattice $\Lambda$, $\epsilon>0$, $r \geq
  \eta_{\epsilon}(\Lambda)$, and $\vecc \in \R^n$, we have $\rho_{r}(\Lambda + \vecc) \in [
  \tfrac{1-\epsilon}{1+\epsilon}, 1 ] \cdot \rho_{r}(\Lambda)$.
\end{lemma}

\begin{lemma}[{{\cite[Claim~3.9]{DBLP:journals/jacm/Regev09}}}]
\label{le:discrete_plus_cont_Gauss}
Let~$\Lambda$ be an $n$-dimensional lattice, let~$\vec{u}\in \R^n$ be
arbitrary, let $r, s > 0$ and let $t=\sqrt{r^2+ s^2}$. Assume that
$rs/t=1/\sqrt{1/r^2+1/s^2} \geq \eta_{\eps} (\Lambda)$ for some~$\eps <
1/2$. Consider the continuous distribution~$Y$ on~$\R^n$ obtained by
sampling from~$D_{\Lambda+\vec{u},r}$ and then adding a noise vector taken
from~$D_s$. Then, the statistical distance between~$Y$ and~$D_t$ is at
most~$4\eps$.
\end{lemma}

\begin{lemma}[{{\cite[Corollary~3.10]{DBLP:journals/jacm/Regev09}}}]
\label{lem:smoothip}
Let~$\Lambda$ be an $n$-dimensional lattice, let~$\vc{u}, \vc{z}\in \R^n$ be arbitrary, and let $r, \alpha > 0$. Assume that
$\parens{1/r^2+\left(\|{\vc{z}}\|/\alpha\right)^2}^{-1/2} \geq \eta_{\eps} (\Lambda)$ for some~$\eps <
1/2$. Then the distribution of $\langle\vc{z},\vc{v}\rangle+e$ where $\vc{v} \getsr D_{\Lambda+\vc{u}, r}$ and $e \getsr D_{\alpha}$, is within statistical distance $4\eps$ of $D_{\beta}$ for $\beta = \sqrt{(r \|\vc{z}\|)^2+\alpha^2}$.
\end{lemma}

\begin{lemma}[{{Special case of~\cite[Theorem 3.1]{DBLP:conf/crypto/Peikert10}}}]
\label{lem:peikertdiscretegauss}
Let $\Lambda$ be a lattice and $r,s>0$ be such that $s \ge \eta_\eps(\Lambda)$ for some $\eps \le 1/2$.
Then if we choose $\vecx$ from the continuous Gaussian $D_{r}$ and then choose $\vecy$ from the
discrete Gaussian $D_{\Lambda-\vecx,s}$ then $\vecx+\vecy$ is within statistical distance $8 \eps$ of
the discrete Gaussian $D_{\Lambda,(r^2+s^2)^{1/2}}$.
\end{lemma}

\subsection{Learning with Errors}
\label{sec:computational_problems}

For integers $n,q \ge 1$, an integer vector $\vec{s}
\in \Z^n$, and a probability distribution $\phi$ on~$\R$, let
$A_{q,\vecs, \phi}$ be the distribution over $\T_q^n \times \T$
obtained by choosing $\veca \in \T_q^n$ uniformly at random and an
error term $e$ from $\phi$, and outputting the pair $(\veca, b =
\inner{\veca, \vecs}+e) \in \T_q^n \times \T$.

\begin{definition}
  \label{def:lwe}
  For integers $n,q \ge 1$, an error distribution $\phi$
  over $\R$, and a distribution $\calD$ over $\Z^n$, the (average-case) decision variant of the $\lwe$
  problem, denoted $\lwe_{n,q,\phi}(\calD)$, is to distinguish 
	given arbitrarily many independent samples,
  the uniform distribution over $\T_q^n \times \T$ from
  $A_{q,\vec{s},\phi}$ for a fixed $\vec{s}$ sampled from $\calD$. The variant where the algorithm only gets a bounded number of samples $m\in\bbN$ is denoted $\lwe_{n,m,q,\phi}(\calD)$.
\end{definition}

Notice that the distribution $A_{q,\vec{s},\phi}$ only depends on $\vec{s} \bmod q$, and so one can assume
without loss of generality that $\vec{s} \in \{0,\ldots,q-1\}^n$. Moreover, using a standard
random self-reduction, for any distribution over secrets $\calD$, one can reduce $\lwe_{n,q,\phi}(\calD)$ to
$\lwe_{n,q,\phi}(U(\{0,\ldots,q-1\}^n))$, and we will occasionally use $\lwe_{n,q,\phi}$ to denote the latter
(as is common in previous work). 
When the noise is a Gaussian with parameter $\alpha>0$, i.e., $\phi=D_\alpha$,
we use the shorthand $\lwe_{n,q,\alpha}(\calD)$.
Since the case when $\calD$ is uniform over $\{0,1\}^n$ plays an important
role in this paper, we will denote it by $\zolwe_{n,q,\phi}$ (and by~$\zolwe_{n,m,q,\phi}$ when the algorithm only gets~$m$ samples). Finally, as we show in the following lemma, 
one can efficiently reduce $\lwe$ to the case in which the secret
is distributed according to the (discretized) error distribution and is hence somewhat short. This latter
form of \lwe, known as the ``normal form,'' was first shown hard in~\cite{DBLP:conf/crypto/ApplebaumCPS09}
for the case of prime $q$.
Here we observe that the proof extends to non-prime $q$, the new technical ingredient being Claim~\ref{clm:invertiblemodq} below.

\begin{lemma}\label{lem:normalform}
  For any $q \ge 25$, $n,m \geq 1$, $\alpha>0$, $\eps < 1/2$ and~$s \geq
  \sqrt{\ln(2n(1+1/\epsilon)/\pi)}/q$, there is an efficient
  (transformation) reduction from~$\lwe_{n,m,q,\alpha}$ to
  $\lwe_{n,m',q,\alpha}(\calD)$
  where~$m' = m - (16n + 4\ln \ln q)$ and $\calD=D_{\Z^n,q(\alpha^2+s^2)^{1/2}}$, that 
	turns advantage $\zeta$ into an advantage of at least $(\zeta-8\eps)/4$.
	In particular, assuming $\alpha \geq \sqrt{\ln(2n(1+1/\epsilon)/\pi)}/q$, we can take $s=\alpha$,
	in which case $\calD=D_{\Z^n,\sqrt{2}q\alpha}$.
\end{lemma}

\begin{proof}
  Consider the first $16n + 4\ln\ln q$ samples $(\vec{a},b)$.
	Using Claim~\ref{clm:invertiblemodq}, with probability at least
	$1-2e^{-1} \ge 1/4$, we can efficiently find a subsequence of the samples
	such that the matrix $A_0 \in \Z_q^{n \times n}$ whose columns
	are formed by the $\vec{a}$ in the subset (scaled up by $q$)
	has an inverse $A_0^{-1} \in \Z_q^{n \times n}$ modulo $q$. If we cannot find such a subsequence, we abort. 
	Let $\vec{b}_0 \in \T^n$ be the vector formed by the corresponding $b$ in the subsequence.
	Let also $\vec{b}'_0 \in \T_q^n$ be $\vec{b}_0 + \vec{x}$ where $\vec{x}$ is
	chosen from $D_{q^{-1}\Z^n - \vec{b}_0,s}$. 
	(Notice that the coset $q^{-1}\Z^n - \vecb_0$ is well defined
because $\vecb_0$ is a coset of $\Z^{n} \subseteq q^{-1}\Z^n$.)
  From each of the remaining $m'$ samples $(\vec{a},b) \in \T_q^n \times \T$ we produce a pair
	\[
	  \big(\vec{a}' = A_0^{-1} \vec{a}, b' = b - \inner{A_0^{-1} \cdot q \vec{a}, \vec{b}'_0}\big) \in \T_q^n \times \T.
	\]
	We then apply the given $\lwe$ oracle to the resulting $m'$ pairs and output its result. 
		
	We now analyze the reduction. First notice that the construction of $A_0$ depends
	only on the $\vec{a}$ component of the input samples, and hence the probability of 
	finding it is the same in case the input is uniform and in case it consists of $\lwe$ samples.
	It therefore suffices in the following to show that there is a distinguishing gap conditioned
	on successfully finding an $A_0$. 
	To that end, first observe that if the input samples $(\veca,b)$ are uniform in $\T_q^n \times \T$
	then so are the output samples $(\veca',b')$.
	Next consider the case that the input samples are distributed according to $A_{q,\vec{s},D_{\alpha}}$
	for some $\vec{s} \in \Z^n$. Then since~$s \geq \eta_{\epsilon}(q^{-1}\Z)$ by Lemma~\ref{lem:boundonsmoothing},
	using Lemma~\ref{lem:peikertdiscretegauss} we get that $\vec{b}_0' = q^{-1} A_0^T \vec{s} + \vece_0$
	where $\vece_0$ is distributed within statistical distance $8\eps$ from $D_{q^{-1}\Z^n, (\alpha^2+s^2)^{1/2}}$.
	Therefore, for each output sample $(\veca',b')$ we have
	\[
	   b' = b - \inner{A_0^{-1} \cdot q \vec{a}, \vec{b}'_0} 
		    = \inner{\veca,\vecs}+ e - \inner{\veca, \vecs} - \inner{A_0^{-1} q \veca, \vece_0}
				= \inner{-q \vec{e}_0, \veca'}+e,
	\]
	where $e$ is an independent error from $D_\alpha$. Therefore, the output samples
	are distributed according to $A_{q,-q\vec{e}_0,D_{\alpha}}$, completing the proof.
\end{proof}

\begin{claim}\label{clm:invertiblemodq}
For any $q \ge 25$, $n \ge 1$, and $t_1 \ge 4, t_2 \ge 1$, given a sequence of $t_1 n+t_2 \ln\ln q$
vectors $\vec{a}_1,\vec{a}_2,\ldots$ chosen uniformly and independently from $\Z_q^n$,
except with probability $e^{-t_1 n /16} + e^{-t_2/4}$, there exists a subsequence 
of $n$ vectors such that the $n \times n$ matrix they form is invertible modulo $q$. 
Moreover, such a subsequence can be found efficiently. 
\end{claim}
\begin{proof}
We consider the following procedure. 
Let $k$ be a counter, initialized to $0$, indicating the number of vectors currently 
in the subsequence, and let $A \in \Z_q^{n \times k}$ be the matrix whose columns are formed by the current
subsequence. We also maintain a unimodular matrix~$U \in \Z^{n \times n}$, 
initially set to the identity, satisfying the invariant that~$U \cdot A \in \Z_q^{n \times k}$ 
has the following form:
  its top $k \times k$ submatrix is upper triangular with each
  diagonal coefficient coprime with~$q$; its bottom $(n-k) \times
  k$ submatrix is zero.
The procedure considers the vectors $\vec{a}_i$ one by one.
For each vector $\vec{a}$, if it is such that the gcd of the last~$n-k$ entries of~$U
  \vec{a}$, call it $g$, is coprime with~$q$, then it does the following:
	it adds $\vec{a}$ to the subsequence, 
	computes (using, say, the extended GCD algorithm) a unimodular matrix $V$ that acts as identity on the first $k$ coordinates 
	and for which the last $n-k$ coordinates of $VU\vec{a}$ are $(g,0,\ldots,0)$,
	replaces $U$ with $VU$, and increments $k$. 

It is easy to see that the procedure's output is correct if it reaches $k=n$.
It therefore suffices to analyze the probability that this event happens. 
For this we use the following two facts to handle the cases~$k<n-1$ and~$k = n-1$, respectively. First, the probability that
the gcd of two uniformly random numbers modulo $q$ is coprime with $q$ is 
\[
\prod_{p|q,~p~\text{prime}} (1-p^{-2}) \ge \prod_{p~\text{prime}} (1-p^{-2}) = \zeta(2)^{-1} \approx 0.61,
\]
where $\zeta$ is the Riemann zeta function. 
Second, the probability that one uniformly random number modulo $q$ is coprime with $q$
is $\varphi(q)/q$, where $\varphi$ is Euler's totient function. By~\cite[Theorem 8.8.7]{BachShallit}, 
this probability is at least $(e^\gamma \ln \ln q + 3/(\ln\ln q))^{-1}$ where $\gamma$ is Euler's constant, which for $q \ge 25$
is at least $(4 \ln \ln q)^{-1}$. 

Using the (multiplicative) Chernoff bound, the first fact, and the fact that $U\vec{a}$ is uniform in $\Z_q^n$ since $U$ is unimodular, we see that the probability that $k < n-1$ after considering $t_1 n$ vector is at most $e^{-t_1 n /16}$. Moreover, once $k=n-1$, using the second fact we get that the probability that after considering $t_2 \ln \ln q$ additional vectors we still have $k=n-1$ is at most $e^{-t_2/4}$.
\end{proof}

\myparagraph{Unknown (Bounded) Noise Rate}
We also consider a variant of $\lwe$ in which the amount of noise is some unknown $\beta \le \alpha$ (as opposed
to exactly $\alpha$), with $\beta$ possibly depending on the secret $\vecs$. As the following lemma shows, this 
does not make the problem significantly harder.

\begin{definition}\label{def:lelwe}
For integers $n,q \ge 1$ and $\alpha \in (0,1)$, $\lwe_{n,q,\le\alpha}$ is the problem of solving $\lwe_{n,q,\beta}$ for any $\beta = \beta(\vc{s}) \le \alpha$.
\end{definition}

\begin{lemma}\label{lem:lelwe}
Let $\cA$ be an algorithm for $\lwe_{n, m, q, \alpha}$ with advantage at least $\epsilon>0$. Then there exists an algorithm $\cB$ for $\lwe_{n,m',q,\le\alpha}$ using oracle access to~$\cA$ and with advantage at least $1/3$, where both $m'$ and its running time are $\poly(m,1/\epsilon, n, \log q)$.
\end{lemma}

The proof is standard (see, e.g., \cite[Lemma 3.7]{DBLP:journals/jacm/Regev09} for
the analogous statement for the search version of $\lwe$). The idea is to use Chernoff bound to estimate $\cA$'s success probability on the uniform distribution, and then add noise in small increments to our given distribution and estimate $\cA$'s behavior on the resulting distributions. If there is a gap between any of these and the uniform behavior, the input distribution is deemed non-uniform. The full proof is omitted.

\znote{Commented out proof is down here. Fix or remove.
}

\myparagraph{Relation to Lattice Problems}
Regev~\cite{DBLP:journals/jacm/Regev09} and Peikert~\cite{DBLP:conf/stoc/Peikert09} showed quantum and classical reductions (respectively) from the worst-case hardness of the $\gapsvp$ problem to the \onote{added:}search version of $\lwe$. 
(We note that the quantum reduction in~\cite{DBLP:journals/jacm/Regev09} also shows a reduction from $\sivp$.)
As mentioned in the introduction, the classical reduction only works when the modulus $q$ is exponential
in the dimension~$n$.
This is summarized in the following theorem, which is derived from
{{\cite[Theorem~3.1]{DBLP:journals/jacm/Regev09}}} and
{{\cite[Theorem~3.1]{DBLP:conf/stoc/Peikert09}}}.

\begin{theorem}\label{thm:worstcaseaveragecase}
Let $n, q \ge 1$ be integers and let $\alpha \in (0,1)$ be such that $\alpha q \ge 2 \sqrt{n}$. Then there exists a quantum reduction from worst-case $n$-dimensional $\gapsvp_{\widetilde{O}(n/\alpha)}$ to $\lwe_{n,q,\alpha}$. If in addition $q \ge 2^{n/2}$ then there is also a classical reduction between those problems. 
\end{theorem}

In order to obtain hardness of the \emph{decision} version of \lwe, which is the one we consider throughout the paper, one employs a search-to-decision reduction. Several such reductions appear in the literature (e.g.,~\cite{DBLP:journals/jacm/Regev09,DBLP:conf/stoc/Peikert09,DBLP:conf/eurocrypt/MicciancioP12}). The most recent reduction by Micciancio and Peikert~\cite{DBLP:conf/eurocrypt/MicciancioP12}, which essentially subsumes all previous reductions, 
requires the modulus~$q$ to be smooth. Below we give the special case when the modulus is a power of~$2$, which suffices for our purposes. It follows from our results that (decision) $\lwe$ is hard not just for a smooth modulus~$q$, as follows from~\cite{DBLP:conf/eurocrypt/MicciancioP12}, but actually for all moduli~$q$, including prime moduli, with only a small deterioration in the noise
(see Corollary~\ref{cor:mod-reduction-2}).

\begin{theorem}[{{Special case of~\cite[Theorem 3.1]{DBLP:conf/eurocrypt/MicciancioP12}}}]\label{thm:searchtodecisionmicpei}
Let $q$ be a power of $2$, and $\alpha$ satisfy $1/q < \alpha < 1/\omega(\sqrt{\log n})$. Then there exists an efficient reduction from search $\lwe_{n,q,\alpha}$
to (decision) $\lwe_{n,q,\alpha'}$ for $\alpha'=\alpha \cdot \omega(\log n)$.
\end{theorem}

%%% Local Variables:
%%% mode: latex
%%% TeX-master: "lwehard"
%%% End:

%\input{extlwe}

%\input{modreduction}

\section{Modulus-Dimension Switching}
\label{sec:modexpansion}

The main results of this section are Corollaries~\ref{cor:mod-reduction} and~\ref{cor:mod-dim-tradeoff} below.
Both are special cases of the following technical theorem.  
%\dnote{I added this and moved from~$B$-bounded to~$(B,\delta)$-bounded} 
We
say that a distribution~$\calD$ over~$\Z^{n}$ is $(B,\delta)$-bounded
for some reals~$B,\delta\geq 0$ if the probability that~$\vec{x} \gets
\calD$ has norm greater than $B$ is at most $\delta$. 

\begin{theorem}
  \label{thm:mod-dim-switch}
  Let $m,n, n',q,q' \geq 1$ be integers, let
  $\matG \in \Z^{n' \times n}$ be such that the lattice $\Lambda =
  \frac{1}{q'} \matG^{T} \Z^{n'} + \Z^{n}$ has a known basis $\matB$,
  and let $\calD$ be an arbitrary $(B,\delta)$-bounded distribution over
  $\Z^{n}$.  Let $\alpha, \beta > 0$ and~$\eps \in (0,1/2)$ satisfy 
	\[\beta^{2} \geq
  \alpha^{2} + (4/\pi) \ln(2n(1+1/\eps)) \cdot (\max\set{q^{-1}, \length{\gs{\matB}}} \cdot B)^{2}.\]  
	Then there is an efficient (transformation) reduction
  from~$\lwe_{n,m,q,\le \alpha}(\calD)$ to $\lwe_{n',m,q',\leq \beta}(\matG
  \cdot \calD)$ that reduces the advantage by at most $\delta + 14 \eps m$.
\end{theorem}

Here we use the notation
$\|\widetilde \matB\|$ from Lemma~\ref{lem:gpv}. We also note that if needed, the distribution on secrets produced
by the reduction can always be turned into the uniform distribution on $\Z_{q'}^{n'}$, as mentioned after Definition~\ref{def:lwe}.
Also, we recall that there exists an elementary reduction from $\lwe_{n',q',\leq \beta}$ 
to~$\lwe_{n',q',\beta}$ (see Lemma~\ref{lem:lelwe}).

Here we state two important corollaries of the theorem.  The first
corresponds to just modulus reduction (the \lwe dimension is
preserved), and is obtained by letting $n'=n$, $\matG=\matI$ be the
$n$-dimensional identity matrix, and $\matB = \matI/q'$.  
For example, we can take $q \ge q' \ge \sqrt{2\ln(2n(1+1/\eps))} \cdot
(B/\alpha)$ and $\beta=\sqrt{2} \alpha$, which corresponds to reducing
an arbitrary modulus to almost~$B/\alpha$, while increasing the
initial error rate~$\alpha$ by just a small constant factor.

\begin{corollary}
  \label{cor:mod-reduction}
  For any $m,n \geq 1$, $q \geq q' \geq 1$, $(B,\delta)$-bounded distribution~$\calD$ over~$\Z^{n}$, $\alpha,\beta>0$ and~$\eps \in (0,1/2)$ such that
  \[
	\beta^{2} \geq \alpha^{2} + (4/\pi) \ln(2n(1+1/\eps))
  \cdot (B/q')^2,
	\]
	there is an efficient reduction from
  $\lwe_{n,m,q,\le \alpha}(\calD)$ to 
	$\lwe_{n,m,q',\le \beta}(\calD)$ that reduces the advantage by at most $\delta + 14 \eps m$.
\end{corollary}

In particular, by using the normal form of \lwe (Lemma~\ref{lem:normalform}), in which the secret has distribution $\calD = D_{\Z^{n}, \sqrt{2} \alpha q}$, we can switch to a power-of-2 modulus with only a small loss in the noise rate, as described in the following
corollary. Together with the known search-to-decision reduction (Theorem~\ref{thm:searchtodecisionmicpei}), this extends the known hardness
of (decision) $\lwe$ to \emph{any} modulus $q$.
Here we use that $\calD = D_{\Z^{n}, r}$
is $(C r \sqrt{n \log(n/\delta)}, \delta)$-bounded for some universal constant $C>0$, which
follows by taking union bound over the $n$ coordinates.
(Alternatively, one could use that it is~$(r\sqrt{n},2^{-n})$-bounded, as follows
from~\cite[Lemma~1.5]{banaszczyk93:_new}, leading to a slightly tighter statement for large $n$.)

\begin{corollary}
  \label{cor:mod-reduction-2}
  Let $\delta \in (0,1/2)$, $m \ge n \geq 1$, $q' \ge 25$. Let also $q \in [q',2q')$ be the smallest power of $2$ not smaller than $q'$
	and $\alpha \ge \sqrt{\ln(2n(1+16/\delta)/\pi)}/q$.
	There exists an efficient (transformation) reduction from $\lwe_{n,m,q, \alpha}$ to $\lwe_{n,m',q', \le \beta}$
	where $m'=m-(16 n + 4 \ln\ln q)$ and 
	\[
	\beta = C \alpha \sqrt{n} \sqrt{\log(n/\delta)\log(m/\delta)}
	\]
	for some universal constant $C>0$, that turns advantage of $\zeta$ into an advantage of at least $(\zeta-\delta)/4$.
\end{corollary}

  Another corollary illustrates a modulus-dimension tradeoff.  Assume
  $n=kn'$ for some $k \ge 1$, and let~$q' = q^{k}$.  Let~$\vec{G} =
  \vec{I}_{n'} \otimes \vec{g}$, where~$\vec{g} = (1,q,q^2, \ldots,
  q^{k-1})^T \in \Z^k$.  We then have $\Lambda = q^{-k} \vec{G}^T
  \Z^{n'} + \Z^n$.  A basis of~$\Lambda$ is given by
\[
\vec{B}= \vec{I}_{n'} \otimes 
\begin{bmatrix}
q^{-1} & q^{-2} & \cdots & q^{-k} \\
       & q^{-1} & \cdots & q^{1-k} \\
       &        & \ddots & \vdots \\
       &		&		 & q^{-1}
     \end{bmatrix} \in \R^{n \times n};
\]
this is since the column vectors of~$\vec{B}$ belong to~$\Lambda$ and
the determinants match.  Orthogonalizing from left to right, we have
$\gs{\matB} = q^{-1} \matI$ and so $\length{\gs{\matB}}=q^{-1}$.  We
therefore obtain the following corollary, showing that we can trade
off the dimension against the modulus, holding~$n \log q = n' \log q'$
fixed.  For example, letting~$\calD = D_{\Z^{n}, \alpha q}$
(corresponding to a secret in normal form, see Lemma~\ref{lem:normalform}), which is $(\alpha q
\sqrt{n}, 2^{-n})$-bounded, 
the reduction increases the error rate by about
a~$\sqrt{n}$ factor. 

\begin{corollary}
  \label{cor:mod-dim-tradeoff}
  For any $n, m,q \ge 1$, $k \ge 1$ that divides~$n$,  $(B,\delta)$-bounded distribution~$\calD$ over~$\Z^{n}$, 
	$\alpha,\beta>0$, and~$\eps \in (0,1/2)$ such 
that
\[
\beta^{2} \geq \alpha^{2} + (4/\pi)
  \ln(2n(1+1/\eps)) \cdot (B/q)^{2},
\]
  there is an efficient reduction from~$\lwe_{n,m,q,\le \alpha}(\calD)$
  to~$\lwe_{n/k,m,q^k,\le \beta}(\vec{G} \cdot \calD)$ that reduces the advantage by at most $\delta+14 \eps m$, 
where~$\vec{G} = \vec{I}_{n/k} \otimes  (1,q,q^2, \ldots, q^{k-1})^T$.
\end{corollary}

\noindent Theorem~\ref{thm:mod-dim-switch} follows immediately from
the following lemma.

\begin{lemma}
  \label{lem:reverse-main}
  Adopt the notation of Theorem~\ref{thm:mod-dim-switch}, and let 
\[
	r
  \geq \max\set{q^{-1}, \length{\gs{\matB}}} \cdot
  \sqrt{2\ln(2n(1+1/\epsilon))/\pi}. 
	\]
There is an efficient 
  mapping from $\T_{q}^{n} \times \T$ to $\T_{q'}^{n'} \times \T$, which
  has the following properties:
  \begin{itemize}[itemsep=0pt]
  \item If the input is uniformly random, then the output is within
statistical distance $4 \epsilon$ from the uniform distribution.
  \item If the input is distributed according to
    $A_{q,\vecs,D_{\alpha}}$ for some $\vecs \in \Z^{n}$ with
    $\length{\vecs} \leq B$, then the output distribution is within
    statistical distance $10 \epsilon$ from $A_{q',\matG\vecs,D_{\alpha'}}$, where
    $(\alpha')^{2} = \alpha^{2} + r^{2}(\length{\vecs}^{2}+B^{2}) \leq
    \alpha^{2} + 2(rB)^{2}$. 
%\dnote{I think we use the condition $\length{\vecs} \leq B$ only for that last inequality. 
%If this is indeed the case, I would be in favour of removing the condition, and commenting at the end of this item that 
%if~$\length{\vecs} \leq B$, then $\alpha^{2} + r^{2}(\length{\vecs}^{2}+B^{2}) \leq
%    \alpha^{2} + 2(rB)^{2}$}\onote{don't we also need it when we apply Lemma~\ref{lem:smoothip} towards the end of the proof? (I didn't check too carefully)}
%\dnote{Why? we don't need $\vec{s}$ small too apply it, do we?}\onote{yes, I think we do; look carefully at the lemma...}
  \end{itemize}
\end{lemma}

\begin{proof}
  The main idea behind the reduction is to encode $\T_{q}^{n}$ into
  $\T_{q'}^{n'}$, so that the mod-$1$ inner products between vectors
  in $\T_{q}^{n}$ and a short vector $\vecs \in \Z^{n}$, and between
  vectors in $\T_{q'}^{n'}$ and $\matG \vecs \in \Z^{n'}$, are nearly
  equivalent. In a bit more detail, the
  reduction will map its input vector $\veca \in \T_{q}^{n}$ (from the
  given LWE-or-uniform distribution) to a vector $\veca' \in
  \T_{q'}^{n'}$, so that
  \[ \inner{\veca', \matG \vecs} = \inner{\matG^{T} \veca', \vecs}
  \approx \inner{\veca,\vecs} \bmod 1 \] for any (unknown) $\vecs \in
  \Z^{n}$.  To do this, it randomly samples $\veca'$ so that
  $\matG^{T} \veca' \approx \veca \bmod \Z^{n}$, where the
  approximation error will be a discrete Gaussian of parameter $r$.

  % In order to formally describe how the reduction samples $\veca'$, we
  % first define an $n$-dimensional lattice \[ \Lambda = q'^{-1} \cdot
  % \lat(\matG^{T}) + \Z^{n} = q'^{-1} \cdot \lat(\matI_{n'} \otimes
  % \vecg) + \Z^{n}, \] and claim that it has as a basis the
  % upper-triangular matrix \[ \matB =
  % \matI_{n'} \otimes
  % \begin{bmatrix}
  %   (q'_{1})^{-1} & \star & \cdots & \star \\
  %   & (q'_{2})^{-1} & \cdots & \star \\
  %   & & \ddots & \vdots \\
  %   & & & (q'_{k})^{-1}
  % \end{bmatrix} \in \R^{n \times n}, \] where the rightmost column of
  % the above displayed $k$-by-$k$ matrix is $q'^{-1} \vecg$, and the
  % $j$th column is obtained by multiplying the $(j+1)$st column by
  % $q'_{j+1}$ and reducing modulo $1$, for $j=k-1,\ldots, 1$.  Indeed,
  % $\matB$ is a basis of $\Lambda$ because all of its columns are in
  % $\Lambda$ by construction, and because
  % $\det(\matB)=q'^{-n'}=\abs{\Lambda/\Z^{n}}^{-n'}=\det(\Lambda)$.  Also
  % notice that because $\matB$ is upper triangular, its Gram-Schmidt
  % orthogonalization (from left to right) is just the diagonal matrix
  % with $(q'_{i})^{-1}$ in the $i$th diagonal, so $\max_{i}
  % \length{\gs{\vecb_{i}}} \leq q_{\min}^{-1}$ and $r \geq
  % \sqrt{2}\smootheps(\Lambda)$ by~\cite[Theorem
  % 3.1]{DBLP:conf/stoc/GentryPV08}\cnote{Put this in prelims...}.
  
  We can now formally define the reduction, which works as follows.
  On an input pair $(\veca, b) \in \T_{q}^{n} \times \T$, it does the
  following:
  \begin{itemize}
  \item Choose $\vecf \gets D_{\Lambda-\veca,r}$ using Lemma~\ref{lem:gpv} with basis
    $\matB$, and let $\vecv = \veca + \vecf \in
    \Lambda/\Z^{n}$. (The coset $\Lambda-\veca$ is well defined since $\veca =
      \bar{\veca}+\Z^{n}$ is some coset of $\Z^{n} \subseteq \Lambda$.) Choose a uniformly
    random solution $\veca' \in \T_{q'}^{n'}$ to the equation
    $\matG^{T} \veca' = \vecv \bmod \Z^{n}$. This can be done by computing a basis of the solution set $\matG^{T}
      \veca' = \veczero \bmod \Z^{n}$, and adding a uniform element from that set to an arbitrary solution to the equation
    $\matG^{T} \veca' = \vecv \bmod \Z^{n}$.
  \item Choose $e' \gets D_{rB}$ and let $b' = b + e' \in \T$.  
  \item Output $(\veca',b')$.
  \end{itemize}

  We now analyze the reduction.  First, if the distribution of the
  input is uniform, then it suffices to show that $\veca'$ is
  (nearly) uniformly random,
  because both $b$ and~$e'$ are independent of~$\veca'$, and $b \in
  \T$ is uniform.  To prove this claim, notice that it suffices to
  show that the coset $\vecv \in \Lambda/\Z^{n}$ is (nearly) uniformly
  random, because each $\vecv$ has the same number of solutions
  $\veca'$ to $\matG^{T} \veca' = \vecv \bmod \Z^{n}$.  Next, observe that for any $\bar{\veca} \in \T_q^n$ and $\bar{\vecf} \in \Lambda - \bar{\veca}$, we have by Lemma~\ref{lem:smoothing} 
	(using that $r \ge \smootheps(\Lambda)$ by Lemma~\ref{lem:boundonsmoothing}) that
	\begin{align}\label{eq:roundingonelatticeanother}
	  \Pr[\veca = \bar{\veca} \wedge \vecf = \bar{\vecf}] &= q^{-n} \cdot \rho_r(\bar{\vecf}) / \rho_r(\Lambda-\bar{\veca}) \nonumber \\
		& \in C \bracks[\big]{1,\tfrac{1+\epsilon}{1-\epsilon}}\cdot \rho_r(\bar \vecf).
	\end{align}
	where $C=q^{-n}/\rho_{r}(\Lambda)$ is a normalizing
  value that does not depend on $\bar{\veca}$ or $\bar{\vecf}$.
	Therefore, by summing over all $\bar \veca, \bar \vecf$ satisfying $\bar \veca + \bar \vecf = \bar \vecv$, we obtain that 
	for any $\bar{\vecv} \in \Lambda/\Z^{n}$,
  \begin{align*}
    \Pr[\vecv = \bar{\vecv}] 
    &\in C \bracks[\big]{1,\tfrac{1+\epsilon}{1-\epsilon}} \cdot
    \rho_{r}(q^{-1} \Z^{n} + \bar{\vecv}).
  \end{align*}
  Since $r
  \geq \smootheps(q^{-1} \Z^{n})$ (by
  Lemma~\ref{lem:boundonsmoothing}), Lemma~\ref{lem:smoothing} 
implies that~$\Pr[\vecv = \bar{\vecv}] \in \bracks[\big]{\tfrac{1-\epsilon}{1+\epsilon},\tfrac{1+\epsilon}{1-\epsilon}} C'$
for a constant~$C'$ that is independent of~$\bar{\vec{v}}$. By Claim~\ref{clm:separationdist}, this shows that~$\veca'$
is within statistical distance~$1-((1-\eps)/(1+\eps))^2 \le 4\epsilon$ of the uniform distribution.

  It remains to show that the reduction maps $A_{q,\vecs,D_{\alpha}}$
  to $A_{q', \matG \vecs, D_{\beta}}$.  Let the input sample from the
  former distribution be $(\veca,b=\inner{\veca,\vecs}+e)$, where $e
  \gets D_{\alpha}$.  As argued above, the output $\veca'$ is (nearly)
  uniform over~$\T_{q'}^{n'}$.  So condition now on any fixed value $\overline{\veca'} \in \T_{q'}^{n'}$ of $\veca'$,
	and let $\bar \vecv = \matG^{T} \overline{\veca'} \bmod \Z^{n}$. 	
	We have
  \[ b' = \inner{\veca,\vecs} + e + e' = \inner{\overline{\veca'}, \matG \vecs} +
  e + \inner{-\vecf, \vecs} + e' \bmod 1. \] 
	By Claim~\ref{clm:separationdist} and~\eqref{eq:roundingonelatticeanother} (and noting that if $\vecf=\bar \vecf$ then $\veca = \bar \vecv - \bar \vecf \bmod \Z^n$),
	the distribution of $-\vecf$ is within statistical distance~$1-(1-\eps)/(1+\eps) \le 2 \eps$ of
	$D_{q^{-1}\Z^{n}-\bar \vecv, r}$.  By Lemma~\ref{lem:smoothip} (using $r
  \geq \sqrt{2} \smootheps(q^{-1} \Z^{n})$ and~$\length{\vecs} \le B$), the
  distribution of $\inner{-\vecf,\vecs} + e'$ is within statistical 
distance~$6 \eps$ 
  from~$D_{t}$, where $t^{2} = r^{2}(\length{\vecs}^{2}+B^{2})$.  It
  therefore follows that $e+\inner{-\vecf,\vecs} + e'$ is within statistical distance~$6 \eps$ 
from $D_{(t^2+\alpha^2)^{1/2}}$, as required.
\end{proof}

\section{Hardness of LWE with Binary Secret}
\label{sec:shortsecret}

The following is the main theorem of this section. 

%\dnote{I changed~$\alpha, \delta \in (0,1)$ to~$\alpha,\delta >0$}

\begin{theorem}\label{thm:zolwe}
Let~$k, q \ge 1$, and $m \ge n \ge 1$ be integers, and let $\epsilon \in (0,1/2)$, $\alpha, \delta >0$,
be such that $n \ge (k+1) \log_2 q + 2 \log_2 (1/\delta)$,
$\alpha \ge \sqrt{\ln(2n(1+1/\eps))/\pi}/q$.
There exist three (transformation) reductions from $\lwe_{k, m, q, \alpha}$ to $\zolwe_{n, m, q, \le \sqrt{10n} \alpha}$,
such that for any algorithm for the latter problem with advantage $\zeta$, at least one of the reductions produces 
an algorithm for the former problem with advantage at least 
\begin{equation}\label{eq:advantagelossinzolwe}
(\zeta-\delta)/(3m) - 41 \eps/2 - \sum_{p|q,~p~\text{prime}} p^{-k-1}\,.
\end{equation}
\end{theorem}

By combining Theorem~\ref{thm:zolwe} with the reduction in Corollary~\ref{cor:mod-reduction} (and noting that $\{0,1\}^n$ is $(\sqrt{n},0)$ bounded), we can replace the $\zolwe$ problem above with $\zolwe_{n,m,q',\beta}$ for any $q' \ge 1$ and $\xi>0$ where
\[
\beta := \left(10 n \alpha^2 + \frac{4n}{\pi q'^2} \ln(2n(1+1/\xi))\right)^{1/2},
\]
while decreasing the advantage in~\eqref{eq:advantagelossinzolwe} by $14 \xi m$. Recalling that 
$\lwe$ of dimension $k=\sqrt{n}$ and modulus $q=2^{k/2}$ (assume $k$ is even) is known to be 
classically as hard as $\sqrt{n}$-dimensional
lattice problems (Theorems~\ref{thm:worstcaseaveragecase} and~\ref{thm:searchtodecisionmicpei}), 
this gives a formal statement of Theorem~\ref{thm:informal}. The modulus $q'$ can be taken
almost as small as $\sqrt{n}$.

For most purposes the sum over prime factors of $q$ in~\eqref{eq:advantagelossinzolwe} is
negligible. For instance, in deriving the formal statement of Theorem~\ref{thm:informal}
above, we used a $q$ that is a power of $2$, in which case the sum is $2^{-k-1}=2^{-\sqrt{n}-1}$, which
is negligible. If needed, one can improve this by applying the modulus switching reduction (Corollary~\ref{cor:mod-reduction-2})
before applying Theorem~\ref{thm:zolwe} in order to make $q$ prime. (Strictly speaking, one also needs to apply Lemma~\ref{lem:lelwe} to replace the ``unknown noise'' variant of $\lwe$ given by Corollary~\ref{cor:mod-reduction-2} with the fixed noise variant.) This improves the advantage loss
to $q^{-\sqrt{n}-1}$ which is roughly $2^{-n}$. 

In a high level, the proof of the theorem follows by combining three main steps. The first, given in
Section~\ref{sec:firstiserrorless}, reduces $\lwe$ to a variant in which the first
equation is errorless. The second, given in Section~\ref{sec:extlwe}, reduces
the latter to the intermediate problem $\extlwe$, another variant of $\lwe$ in which some information 
on the noise elements is leaked. Finally, in Section~\ref{sec:theshortsecretreduction},
we reduce $\extlwe$ to $\lwe$ with $\{0,1\}$ secret. We note that the first reduction is
relatively standard; it is the other two that we consider as the main contribution
of this section. We now proceed with more details (see also Figure~\ref{fig:summary}). 

\begin{proof}
First, since $m \ge n$, Lemma~\ref{lem:firsterrorlesshard} provides a transformation reduction
from $\lwe_{k, m, q, \alpha}$ to first-is-errorless $\lwe_{k+1, n, q, \alpha}$, while reducing the advantage by at most $2^{-k+1}$. 
Next, Lemma~\ref{lem:redtoextlwe} with $\cZ = \{0,1\}^n$, which is of quality $\xi=2$ by 
Claim~\ref{clm:qualityofzeroone}, reduces the latter problem to $\extlwe_{k+1, n, q,\sqrt{5}\alpha,\{0,1\}^n}$ while 
reducing the advantage by at most~$33 \eps/2$. Then, Lemma~\ref{lem:multiinstanceextlwe} reduces the latter problem
to $\extlwe^m_{k+1, n, q, \sqrt{5} \alpha,\{0,1\}^n}$, while losing a factor of~$m$ in the advantage. 
Finally, Lemma~\ref{lem:zolwe} provides three reductions to $\zolwe_{n, m, q, \le \sqrt{10n} \alpha}$: two from the latter problem, and one from
$\lwe_{k+1, m, q,\sqrt{5n} \alpha}$, guaranteeing that the sum of advantages is at least the original advantage minus $4 m \eps + \delta$.
Together with the trivial reduction from $\lwe_{k, m, q,\alpha}$ to $\lwe_{k+1, m, q, \sqrt{5n} \alpha}$ (which incurs no loss in advantage), this
completes the proof. 
\end{proof}

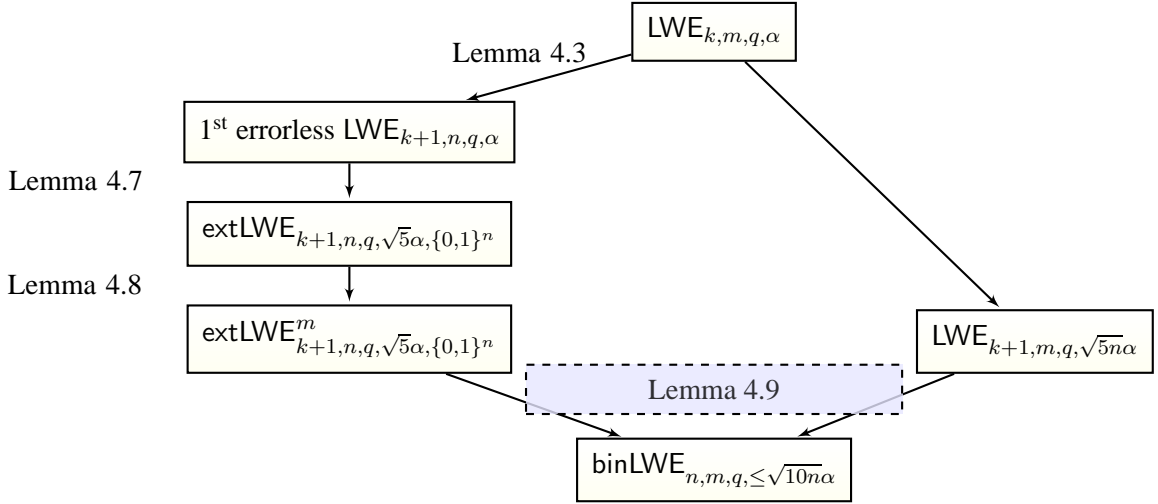
\begin{figure}[ht]%
\centering
\begin{tikzpicture}[node distance=1cm, auto]  
\tikzset{
    mynode/.style={rectangle,draw=black, top color=white, bottom color=yellow!5,thick, inner sep=0.5em, minimum height=2em, text centered},
    myarrow/.style={->, >=latex', shorten >=1pt, thick},
    mylabel/.style={text width=7em, text centered} 
}  
%\node[mynode] (final) {$\zolwe_{n, m',q, \sqrt{10n} \alpha}$};  
%\node[mynode, above=0.5cm of final] (atmost) {$\zolwe_{n, m,q, \le \sqrt{10n} \alpha}$}; 
\node[mynode] (atmost) {$\zolwe_{n, m,q, \le \sqrt{10n} \alpha}$}; 
\node[mynode, above right=1.2cm of atmost] (second) {$\lwe_{k+1, m, q, \sqrt{5n} \alpha}$};
\node[mynode, above left=1.2cm of atmost] (first) {$\extlwe^m_{k+1, n, q,\sqrt{5} \alpha,\{0,1\}^n}$};  
%\node[mynode, right=of second] (third) {discrete $\lwe^n_{k, m, q,\beta}$};
%\node[mynode, above=0.5cm of third] (thirdabove) {discrete $\lwe_{k, m, q,\beta}$};
%\node[mynode, above=0.5cm of thirdabove] (thirdaboveabove) {$\lwe_{k, m, q,\beta}$};
\node[mynode, above=0.5cm of first] (firstabove) {$\extlwe_{k+1, n, q, \sqrt{5} \alpha,\{0,1\}^n}$};
\node[mynode, above=0.5cm of firstabove] (firstaboveabove) {1$^{\mbox{\scriptsize{st}}}$ errorless $\lwe_{k+1, n,q,\alpha}$};
\node[mynode, above=5cm of atmost] (firstaboveaboveabove) {$\lwe_{k, m, q,\alpha}$};

%\draw[myarrow] (atmost)  --  (final);	
\draw[myarrow] (first)  --  (atmost);	
\draw[myarrow] (second)  --  (atmost);	
%\draw[myarrow] (third)  --  (atmost);	
%\draw[myarrow] (thirdabove)  --  (third);	
%\draw[myarrow] (thirdaboveabove)  --  (thirdabove);	
\draw[myarrow] (firstabove)  --  (first);	
\draw[myarrow] (firstaboveabove)  --  (firstabove);	
\draw[myarrow] (firstaboveaboveabove)  --  (firstaboveabove);	
\draw[myarrow] (firstaboveaboveabove)  --  (second);	
%\draw[myarrow] (firstaboveaboveabove)  --  (thirdaboveabove);	

%\node[mylabel, above left=0cm of final] (label1) {Lemma~\ref{lem:lelwe}};  
%\node[mylabel, below left=0cm of first] (label2) {Lemma~\ref{lem:zolwe}};  
\node[mylabel, above left=0cm of first] (label3) {Lemma~\ref{lem:multiinstanceextlwe}};  
\node[mylabel, above left=0cm of firstabove] (label4) {Lemma~\ref{lem:redtoextlwe}};  
\node[mylabel, below left=0cm of firstaboveaboveabove.west] (label5) {Lemma~\ref{lem:firsterrorlesshard}};  
%\node[mylabel, above right=0cm of third] (label6) {Lemma~\ref{lem:multiinstancelwe}};  
%\node[mylabel, above right=0cm of thirdabove] (label7) {Lemma~\ref{lem:conttodiscrete}};  
\node[rectangle,draw=black, thick, inner sep=0.5em, minimum width=5cm, fill=blue!10!white, fill opacity=0.8, text opacity=1, dashed, above=0.3cm of atmost] (emphasize) {Lemma~\ref{lem:zolwe}};
\end{tikzpicture} 
\caption{Summary of reductions used in Theorem~\ref{thm:zolwe}}%
\label{fig:summary}%
\end{figure}

\subsection{First-is-errorless LWE}
\label{sec:firstiserrorless}

We first define a variant of LWE in which the first equation is given without error, and then show in
Lemma~\ref{lem:firsterrorlesshard} that it is still hard. 

\begin{definition}\label{def:firsterrorless}
For integers $n,q \ge 1$ and an error distribution $\phi$ over $\R$, the ``first-is-errorless'' variant of the $\lwe$ problem
is to distinguish between the following two scenarios.
In the first, the first sample is uniform over $\T_q^n \times \T_q^{}$ and the rest are uniform
over $\T_q^n \times \T$. In the second, there is an unknown uniformly distributed $\vecs \in \{0,\ldots,q-1\}^n$,
the first sample we get is from $A_{q,\vec{s},\{0\}}$ (where $\{0\}$ denotes the distribution that is deterministically zero)
and the rest are from $A_{q,\vec{s},\phi}$.
\end{definition}

\begin{lemma}\label{lem:firsterrorlesshard}
For any $n \ge 2$, $m, q \ge 1$, and error distribution $\phi$,
there is an efficient (transformation) reduction from $\lwe_{n-1,m,q,\phi}$ to the first-is-errorless variant of $\lwe_{n,m,q,\phi}$
that reduces the advantage by at most $\sum_p p^{-n}$, with the sum going over all prime factors of $q$.
\end{lemma}
Notice that if $q$ is prime the loss in advantage is at most $q^{-n}$. Alternatively, for any number $q$ we can bound it by
\[
\sum_{k \ge 2} k^{-n} \le 2^{-n} + \int_{2}^\infty t^{-n} \mathrm{d} t \le 2^{-n+2},
\]
which might be good enough when $n$ is large.
\begin{proof}
The reduction starts by choosing a
  vector $\veca'$ uniformly at random from
  $\{0,\ldots,q-1\}^n$. Let $r$ be the greatest common divisor of the coordinates of $\veca'$. 
	If it is not coprime to $q$, we abort. The probability that this
  happens is at most 
\[
 \sum_{p~\text{prime}, \ p|q} p^{-n}.
\]
Assuming we do not abort, we proceed by finding a matrix $\matU \in \Z^{n \times n}$ that is invertible modulo $q$ and
whose leftmost column is $\veca'$.
Such a matrix exists, and can be found efficiently. For instance, using the extended GCD algorithm, we find an $n\times n$ unimodular matrix
$\matR$ such that $\matR \veca' = (r,0,\ldots,0)^T$. Then $\matR^{-1} \cdot \mathrm{diag}(r,1,\ldots,1)$ is the desired matrix.
We also pick a uniform element $s_0 \in \{0,\ldots,q-1\}$. The reduction now proceeds as follows. The first sample it outputs is
$(\veca'/q, s_0/q)$. The remaining samples are produced by taking a sample $(\veca,b)$ from the given oracle, picking
a fresh uniformly random $d \in \T_q$, and outputting $(\matU(d|\veca),b+(s_0 \cdot d))$ with the vertical bar denoting concatenation.
It is easy to verify correctness: given uniform samples, the reduction outputs uniform samples
(with the first sample's $b$ component uniform over $\T_q$), up to statistical distance $2^{-n+1}$;
and given samples from $A_{q,\vec{s},\phi}$, the reduction outputs one sample from $A_{q,\vecs',\{0\}}$
and the remaining samples from $A_{q,\vecs',\phi}$, up to statistical distance $2^{-n+1}$,
where $\vecs' = (\matU^{-1})^T (s_0|\vecs) \bmod q$.
This proves correctness since $\matU$, being invertible modulo $q$, induces a bijection on $\Z_q^n$, and so $\vecs'$ is uniform in $\{0,\ldots,q-1\}^n$.
\end{proof}

\subsection{Extended LWE}
\label{sec:extlwe}

We next define the intermediate problem $\extlwe$. (This definition is
of an easier problem than the one considered in previous
work~\cite{DBLP:conf/pkc/Alperin-SheriffP12}, which makes our hardness result stronger.)

\begin{definition}\label{def:extlwe}
For $n,m,q,t \ge 1$, $\calZ \subseteq \Z^m$, and a distribution $\chi$ over $\frac{1}{q}\Z^m$, the $\extlwe^t_{n,m,q,\chi,\calZ}$ problem is as follows.
The algorithm gets to choose $\vecz \in \calZ$ and then receives a tuple
\[
(\matA,(\vecb_i)_{i \in [t]},(\inner{\vece_i,\vecz})_{i \in [t]}) \in \T_q^{n \times m} \times (\T_q^m)^t \times ({\textstyle \frac 1 q} \Z)^t.
\]
Its goal is to distinguish between the following two cases.
In the first, $\matA \in \T_q^{n \times m}$ is chosen uniformly,
$\vece_i \in {\textstyle \frac 1 q}\Z^m$ are chosen from $\chi$,
and $\vecb_i = \matA^T \vecs_i + \vece_i \bmod 1$ where $\vecs_i \in \{0,\ldots,q-1\}^n$ are
chosen uniformly.
The second case is identical, except that the $\vecb_i$ are chosen uniformly in $\T_q^m$ independently of everything else. 
\end{definition}

When $t=1$, we omit the superscript $t$. Also, when $\chi$ is $D_{q^{-1}\Z^{m},\alpha}$ for some $\alpha>0$, we replace the
subscript $\chi$ by $\alpha$.
We note that a discrete version of $\lwe$ can be defined as a special case of $\extlwe$ by setting $\cZ = \{ 0^m \}$.
We next define a measure of quality of sets $\calZ$.

\begin{definition}\label{def:qualityofhintset}
For a real $\xi>0$ and a set $\calZ \subseteq \Z^m$ we say that $\calZ$ is of \emph{quality} $\xi$
if given any $\vecz \in \calZ$, we can efficiently find a unimodular matrix
$\matU \in \Z^{m \times m}$ such that if $\matU'\in \Z^{m \times (m-1)}$ is the matrix obtained
from $\matU$ by removing its leftmost column then all of the columns of $\matU'$ are orthogonal to $\vecz$
and its largest singular value is at most~$\xi$.
\end{definition}

The idea in this definition is that the columns of $\matU'$ 
form a basis of the lattice of integer points that are orthogonal
to $\vecz$, i.e., the lattice $\{\vecb \in \Z^m:
  \inner{\vecb,\vecz} = 0\}$. The quality measures
	how ``short'' we can make this basis.

\begin{claim}\label{clm:qualityofzeroone}
The set $\calZ=\{0,1\}^m$ is of quality $2$.
\end{claim}
\begin{proof}
Let $\vecz \in \calZ$ and assume without loss of generality that its first $k \ge 1$ coordinates are $1$
and the remaining $m-k$ are $0$. Then consider the upper bidiagonal matrix $\matU$ whose diagonal is all $1$s and whose
diagonal above the main diagonal is $(-1,\ldots,-1,0,\ldots,0)$ with $-1$ appearing $k-1$ times. The matrix is clearly
unimodular and all the columns except the first one are orthogonal to $\vecz$. Moreover, by the triangle inequality,
we can bound the operator norm of $\matU$ by the sum of that of the diagonal $1$ matrix and the off-diagonal matrix, both
of which clearly have norm at most $1$.
\end{proof}

\onote{todo: add more quality bounds such as vectors of bounded norm}

\begin{lemma}\label{lem:redtoextlwe}
Let $\calZ \subseteq \Z^m$ be of quality $\xi>0$. Then for any $n,q \ge 1$, $\eps \in (0,1/2)$, and $\alpha,r  \ge (\ln(2m(1+1/\eps))/\pi)^{1/2}/q$,
there is a (transformation) reduction from the first-is-errorless variant of $\lwe_{n,m,q,\alpha}$ to
$\extlwe_{n,m,q,(\alpha^2 \xi^2 + r^2)^{1/2},\calZ}$ that reduces the advantage by at most $33 \eps/2$. 
\end{lemma}
\begin{proof}
We first describe the reduction.
Assume we are asked to provide samples for some $\vecz \in \calZ$.
We compute a unimodular $\matU \in \Z^{m \times m}$ for $\vecz$ as in Definition~\ref{def:qualityofhintset}, and let $\matU' \in \Z^{m \times (m-1)}$ be
the matrix formed by removing the first column of $\matU$.
We then take $m$ samples from the given distribution, resulting in $(\matA, \vecb) \in \T_q^{n \times m} \times (\T_q^{} \times \T^{m-1})$.
We also sample a vector $\vecf$ from the $m$-dimensional continuous Gaussian distribution
$D_{\alpha(\xi^2 \matI - \matU' \matU'^T)^{1/2}}$,
 which is well defined since $\xi^2 \matI - \matU' \matU'^T$ is a positive semidefinite matrix
by our assumption on $\matU$.
The output of the reduction is the tuple
\begin{equation}\label{eq:outputofextlwereduction}
(\matA'= \matA \matU^T,
  \vecb' + \vecc,
	\inner{\vecz , \vecf + \vecc})
	\in \T_q^{n \times m} \times \T_q^m \times {\textstyle{\frac{1}{q}}}\Z,
\end{equation}
where $\vecb'=\matU \vecb + \vecf$, and $\vecc$ is chosen from the discrete
Gaussian distribution $D_{q^{-1}\Z^m - \vecb',r}$ (using Lemma~\ref{lem:gpv}). 

We now prove the correctness of the reduction. Consider first the case that we get valid LWE equations, i.e.,
$\matA$ is uniform in~$\T_q^{n \times m}$ and $\vecb = \matA^T \vecs + \vece \in \T^m$
where $\vecs \in \{0,\ldots,q-1\}^n$ is uniformly chosen, the first coordinate of $\vece \in \R^m$ is~$0$, 
and the remaining $m-1$ coordinates are chosen from~$D_\alpha$.
Since $\matU$ is unimodular, $\matA' = \matA \matU^T$ is uniformly distributed in $\T_q^{n \times m}$ as required.
From now on we condition on an arbitrary $\matA'$ and analyze the distribution of the remaining two
components of~\eqref{eq:outputofextlwereduction}.
Next,
\[
\vecb' =
\matU \vecb + \vecf =
\matA'^T \vecs + \matU \vece + \vecf.
\]
Since $\matU \vece$ is distributed as a continuous Gaussian $D_{\alpha \matU'}$,
the vector~$\matU \vece + \vecf$ is a distributed as a \emph{spherical} continuous 
Gaussian~$D_{\alpha \xi}$.
Moreover, since $\matA'^T \vecs \in \T_q^m$, the coset $q^{-1}\Z^m - \vecb'$ is identical to
$q^{-1}\Z^m - (\matU \vece + \vecf)$, so we can see $\vecc$ as being chosen from
$D_{q^{-1}\Z^m - (\matU \vece + \vecf),r}$. Therefore, by Lemma~\ref{lem:peikertdiscretegauss}
and using that $r \ge \eta_\eps(q^{-1}\Z^m)$ 
by Lemma~\ref{lem:boundonsmoothing},
the distribution of $\matU \vece + \vecf + \vecc$ is within statistical distance $8 \eps$ 
of $D_{q^{-1}\Z^m,(\alpha^2 \xi^2 + r^2)^{1/2}}$. This shows that the second
component in~\eqref{eq:outputofextlwereduction} is also distributed correctly.
Finally, for the third component, by our assumption on $\matU$ and the fact that the
first coordinate of $\vece$ is zero,
\[
\inner{\vecz, \vecf + \vecc} = \inner{\vecz, \matU\vece + \vecf + \vecc},
\]
and so the third component gives the inner product of the noise with $\vecz$, as desired.

We now consider the case where the input is uniform, i.e.,
that $\matA$ is uniform in~$\T_q^{n\times m}$
and $\vecb$ is independent and uniform in $\T_q \times \T^{m-1}$.
We first observe that by Lemma~\ref{lem:smoothingcontinuous}, since $\alpha \ge \eta_{\eps/m}(q^{-1}\Z)$ (by Lemma~\ref{lem:boundonsmoothing}),
the distribution of~$(\matA,\vecb)$ is within statistical distance $\eps/2$
of the distribution of $(\matA, \vece' + \vece)$ where $\vece'$ is
chosen uniformly in~$\T_q^m$, the first coordinate of $\vece$ is zero,
and its remaining $m-1$ coordinates are chosen independently from~$D_\alpha$.
So from now on assume our input is $(\matA,\vece'+\vece)$.
The first component of~\eqref{eq:outputofextlwereduction} is uniform
in $\T_q^{n \times m}$ as before, and moreover, it is clearly independent
of the other two. Moreover, since $\vecb' = \matU \vece' + \matU \vece + \vecf$ and $\matU \vece' \in \T_q^m$,
the coset $q^{-1}\Z^m - \vecb'$ is identical to
$q^{-1}\Z^m - (\matU \vece + \vecf)$, and so $\vecc$ is distributed
identically to the case of a valid LWE equation, and in particular
is independent of $\vece'$. This establishes that the third component
of~\eqref{eq:outputofextlwereduction} is correctly distributed;
moreover, since $\vece'$ is independent of the first and third components,
and $\matU \vece'$ is uniform in $\T_q^m$ (since $\matU$ is unimodular),
we get that the second component is uniform and independent of the other two,
as desired.
\end{proof}

We end this section by stating the standard reduction to the multi-secret ($t \ge 1$) case of extended LWE. 

%The proof is by a straightforward hybrid argument and is omitted.

\onote{removed $\calD$ from the following lemma since we never  use any distribution on secrets but the uniform one}

\begin{lemma}\label{lem:multiinstanceextlwe}
Let $n, m, q, \chi, \cZ$ be as in Definition~\ref{def:extlwe} with $\chi$ efficiently sampleable, and let $t \ge 1$ be an integer. 
Then there is an efficient (transformation) reduction from $\extlwe_{n,m,q,\chi,\cZ}$ to $\extlwe^t_{n,m,q,\chi,\cZ}$ that reduces the advantage by 
a factor of~$t$.
%\dnote{I changed~$\phi$ to $\chi$, for consistency with 
%definition~\ref{def:extlwe}.} 
\end{lemma}

The proof is by a standard hybrid argument. We bring it here for the sake of completeness. We note that the distribution of the secret vector $\vc{s}$ needs to be sampleable but otherwise it plays no role in the proof. The lemma therefore naturally extends to any (sampleable) distribution of $\vc{s}$.

\begin{proof}
Let $\cA$ be an algorithm for $\extlwe^t_{n,m,q,\chi,\cZ}$, let $\vc{z}$ be the vector output by $\cA$ in the first step (note that this is a random variable) and let $H_i$ denote the distribution
\[
\left(\mx{A}, \{\vc{b}_1, \ldots, \vc{b}_i, \vc{u}_{i+1}, \dots, \vc{u}_t\}, \vc{z}, \{\langle\vc{z}, \vc{e}_i\rangle\}_{i\in[t]} \right)~,\]
%\dnote{I added:}
where~$\vc{u}_{i+1}, \dots, \vc{u}_t$ are sampled independently and uniformly in~$\T_q^m$.
Then by definition $\adv[\cA] = \abs{\Pr[\cA(H_0)]- \Pr[\cA(H_t)]}$.

We now describe an algorithm $\cB$ for $\extlwe_{n,m,q,\chi,\cZ}$: First, $\cB$ runs $\cA$ to obtain $\vc{z}$ and sends it to the challenger as its own $\vc{z}$. Then, given an input $(\mx{A}, \vc{d}, \vc{z}, y)$ for $\extlwe_{n,m,q,\chi,\cZ}$, the distinguisher $\cB$ samples $i^* \getsr [t]$, and in addition $\vc{s}_1, \ldots, \vc{s}_{i^*-1}\getsr \bbZ^n_q$, $\vc{e}_1, \ldots, \vc{e}_{i^*-1}, \vc{e}_{i^*+1}, \ldots, \vc{e}_t \getsr \chi^m$\onote{replaced $\phi$ with $\chi$}, $\vc{u}_{i^*+1}, \ldots, \vc{u}_t \getsr \bbT_q^m$. It sets $\vc{b}_i = \mx{A}^T\cdot\vc{s}_i+\vc{e}_i\pmod{1}$, and sends the following to $\cA$:
\[
\left(\mx{A}, \{ \vc{b}_1, \ldots, \vc{b}_{i^*-1}, \vc{d}, \vc{u}_{i^*+1}, \ldots, \vc{u}_t\}, \vc{z}, \{\langle\vc{z},\vc{e}_1\rangle, \ldots,\langle\vc{z}, \vc{e}_{i^*-1}\rangle, y, \langle\vc{z}, \vc{e}_{i^*+1}\rangle, \ldots, \langle\vc{z}, \vc{e}_t\rangle  \} \right)~.
\]
Finally, $\cB$ outputs the same output as $\cA$ did.

Note that when the input to $\cB$ is distributed as $P_0=(\mx{A}, \vc{b}, \vc{z}, \vc{z}^T \cdot \vc{e})$ with~$\vc{b} = \mx{A}^T\cdot\vc{s}+\vc{e}\pmod{1}$, then $\cB$ feeds $\cA$ with exactly the distribution $H_{i^*}$. On the other hand, if the input to $\cB$ is $P_1 = (\mx{A}, \vc{u}, \vc{z}, \vc{z}^T \cdot \vc{e})$ with~$\vc{u} \getsr \bbT_q^m$, then $\cB$ feeds $\cA$ with~$H_{i^*-1}$.

Since $i^*$ is uniform in $[t]$, we get that
\begin{eqnarray*}
t \, \adv[\cB] & = & t \, \abs{\Pr[\cB(P_0)] - \Pr[\cB(P_1)]}\\
  & = &
	\abs[\bigg]{{\sum_{i^*\in[t]} \Pr[\cA(H_{i^*})] - \sum_{i^*\in[t]} \Pr[\cA(H_{i^*-1})]}} \\
	& = &
	\abs{\Pr[\cA(H_{t})]- \Pr[\cA(H_{0})] } \\
	& = & \adv[\cA]~,
\end{eqnarray*}
and the result follows.
\end{proof}

%%% Local Variables:
%%% mode: latex
%%% TeX-master: "lwehard"
%%% End:

\subsection{Reducing to binary secret}\label{sec:theshortsecretreduction}

\begin{lemma}\label{lem:zolwe}
Let $k,n,m,q \in \bbN$, $\epsilon \in (0,1/2)$, and $\delta,\alpha,\beta,\gamma>0$ be such that $n \ge k \log_2 q + 2 \log_2 (1/\delta)$, $\beta \ge \sqrt{2 \ln(2n(1+1/\eps))/\pi}/q$, $\alpha = \sqrt{2n} \beta$, $\gamma = \sqrt{n}\beta$. 
Then there exist three efficient (transformation) reductions to $\zolwe_{n,m,q,\le\alpha}$ from 
$\extlwe^m_{k,n,q,\beta,\{0,1\}^n}$, $\lwe_{k,m,q,\gamma}$, and  $\extlwe^m_{k,n,q,\beta,\{0^n\}}$,
such that if $\cB_1$, $\cB_2$, and $\cB_3$ are the algorithms obtained by applying these reductions (respectively) 
to an algorithm $\cA$, then
\[
\adv[\cA] \le \adv[\cB_1]+\adv[\cB_2]+\adv[\cB_3]+4m\epsilon + \delta \,.
\]
\end{lemma}

Pointing out the trivial (transformation) reduction from $\extlwe^m_{k,n,q,\beta,\{0,1\}^n}$ to $\extlwe^m_{k,n,q,\beta,\{0^n\}}$, the lemma implies the hardness of $\zolwe_{n,m,q,\le\alpha}$ based on the hardness of $\extlwe^m_{k,n,q,\beta,\{0,1\}^n}$ and $\lwe_{k,m,q,\gamma}$.

We note that our proof is actually more general, and holds for any binary distribution of
min-entropy at least $k \log_2 q + 2 \log_2 (1/\delta)$, and not just a uniform binary secret as in the definition of $\zolwe$.

\def\veh{\hat{\vc{e}}}

\begin{proof}
The proof follows by a sequence of hybrids. 
Let $k,n,m,q,\eps, \alpha, \beta, \gamma$ be as in the lemma statement. 
We consider $\vc{z} \getsr \binset^n$ and $\vc{e} \getsr D_{\alpha'}^m$ for $\alpha' = \sqrt{\beta^2\|\vc{z}\|^2+ \gamma^2} \le \sqrt{2n}\beta=\alpha$. In addition, we let $\mx{A}\getsr \bbT_q^{n \times m}$, $\vc{u} \getsr \bbT^m$, and define $\vc{b} \getsd \mx{A}^T\cdot\vc{z}+\vc{e}\pmod{1}$. We consider an algorithm~$\cA$ that distinguishes between $(\mx{A}, \vc{b})$ and $(\mx{A}, \vc{u})$.
%\dnote{I modified the last sentence a bit}

We let $H_0$ denote the distribution $(\mx{A}, \vc{b})$ and $H_1$ the distribution 
\[ H_1 = (\mx{A}, \mx{A}^T\vc{z}-\mx{N}^T\vc{z}+\veh \bmod 1 ),\]
where $\mx{N}\getsr D^{n \times m}_{q^{-1}\bbZ, \beta}$ and $\veh \getsr D_{\gamma}^m$.
Using $\| \vc{z} \| \le \sqrt{n}$ and that $\beta \ge \sqrt{2} \eta_\eps(\Z^n)/q$ (by Lemma~\ref{lem:boundonsmoothing}), it follows by Lemma~\ref{lem:smoothip} that the statistical distance between $-\mx{N}^T\vc{z}+\veh$ and $D^m_{\alpha'}$ is at most $4 m \epsilon$.
It thus follows that
\begin{equation}\label{eq:h0}
\abs{\Pr[\cA(H_0)] - \Pr[\cA(H_1)]} \le 4m\epsilon~.
\end{equation}

We define a distribution $H_2$ as follows. Let $\mx{B} \getsr \bbT_q^{k \times m}$ and $\mx{C} \getsr \bbT_q^{k \times n}$.
Let $\mAh \getsd q\mx{C}^T\cdot \mx{B} + \mx{N}\pmod{1}$. Finally,
\[
H_2 = (\mAh, \mAh^T\cdot\vc{z} - \mx{N}^T\vc{z}+\veh) = (\mAh, q\mx{B}^T\cdot \mx{C}\cdot\vc{z}+\veh)~.
\]

We now argue that there exists an adversary $\cB_1$ for problem $\extlwe^m_{k,n,q,\beta,\{0,1\}^n}$, such that
\begin{equation}\label{eq:h1}
\adv[\cB_1] = \abs{\Pr[\cA(H_1)]-\Pr[\cA(H_2)]}~.
\end{equation}
This is because $H_1, H_2$ can be viewed as applying the same efficient transformation on the distributions $(\mx{C}, \mx{A}, \mx{N}^T\vc{z})$ and $(\mx{C}, \mAh, \mx{N}^T\vc{z})$ respectively. Since distinguishing the latter distributions is exactly the $\extlwe^m_{k,n,q,\beta,\{0,1\}^n}$ problem (where the columns of $q\cdot\mx{B}$ are interpreted as the $m$ secret vectors), the distinguisher $\cB_1$ follows by first applying the aforementioned transformation and then applying~$\cA$.

For the next hybrid, we define $H_3 = (\mAh, \mx{B}^T\cdot \vc{s}+\veh)$, for $\vc{s} \getsr \bbZ_q^k$.
It follows that
\begin{equation}\label{eq:h2}
\abs{\Pr[\cA(H_2)]- \Pr[\cA(H_3)]} \le \delta
\end{equation}
by the leftover hash lemma (see Lemma~\ref{lem:lhl}), since $H_2, H_3$ can be derived from $(\mx{C}, q\mx{C}\cdot\vc{z})$ and $(\mx{C}, \vc{s})$ respectively, whose statistical distance is at most $\delta$.

Our next hybrid makes the second component uniform: $H_4=(\mAh, \vc{u})$. There exists an algorithm $\cB_2$ for $\lwe_{k,m,q,\gamma}$ such that
\begin{equation}\label{eq:h3}
\adv[\cB_2] = \abs{\Pr[\cA(H_3)]- \Pr[\cA(H_4)]}~,
\end{equation}
since $H_3, H_4$ can be computed efficiently from $(\mx{B}, \mx{B}^T\vc{s}+\veh), (\mx{B}, \vc{u})$.

Lastly, we change $\mAh$ back to uniform: $H_5 = (\mx{A}, \vc{u})$. There exists an algorithm $\cB_3$ for $\extlwe^m_{k,n,q,\beta,\{0^n\}}$ such that
\begin{equation}\label{eq:h4}
\adv[\cB_3] = \abs{\Pr[\cA(H_4)]- \Pr[\cA(H_5)]}~.
\end{equation}
Eq.~\eqref{eq:h4} is derived very similarly to Eq.~\eqref{eq:h1}: We notice that $H_4$, $H_5$ can be viewed as applying the same efficient transformation on the distributions $(\mx{C}, \mAh)$ and $(\mx{C}, \mx{A})$ respectively. Since distinguishing the latter distributions is exactly the $\extlwe^m_{k,n,q,\beta,\{0^n\}}$ problem (where the columns of $q\cdot\mx{B}$ are interpreted as the $m$ secret vectors), the distinguisher $\cB_3$ follows by first applying the aforementioned transformation and then applying~$\cA$.

Putting together Eq.~\eqref{eq:h0},~\eqref{eq:h1},~\eqref{eq:h2},~\eqref{eq:h3},~\eqref{eq:h4}, the lemma follows.
\end{proof}

%%% Local Variables:
%%% mode: latex
%%% TeX-master: "lwehard"
%%% End:

\section{Exact Gaussian Sampler}
\label{sec:exactgpv}

In this section we prove Lemma~\ref{lem:gpv}. 
As in~\cite{DBLP:conf/stoc/GentryPV08}, the proof consists of two parts. 
In the first we consider the one-dimensional case, and in the second we use
it recursively to sample from arbitrary lattices. Our one-dimensional
sampler is based on rejection sampling, just like the one in~\cite{DBLP:conf/stoc/GentryPV08}.
Unlike~\cite{DBLP:conf/stoc/GentryPV08}, we use the continuous normal distribution as the source distribution
which allows us to avoid truncation, and as a result obtain an exact sample. 
Our second part uses the same recursive routine as in~\cite{DBLP:conf/stoc/GentryPV08},
but adds a rejection sampling step to it in order to take care of the deviation 
of its output from the desired distribution. 

\subsection{The one-dimensional case}

Here we show how to sample from the discrete Gaussian distribution on arbitrary cosets of
one-dimensional lattices. We use a standard rejection sampling procedure (see, e.g.~\cite[Page 117]{DevroyeBook86}
for a very similar procedure).

By scaling, we can restrict without loss of generality to the lattice $\Z$, i.e.,
we consider the task of sampling from $D_{\Z+c,r}$ for a given coset representative $c \in [0,1)$ and parameter $r>0$.
The sampling procedure is as follows. Let $Z_0 = \int_c^\infty \rho_r(x) \mathrm{d} x$, and $Z_1 = \int_{-\infty}^{c-1} \rho_r(x) \mathrm{d} x$.
These two numbers can be computed efficiently by expressing them in terms of the error function. Let $Z=Z_0+Z_1+\rho_r(c)+\rho_r(c-1)$.
The algorithm repeats the following until it outputs an answer: 
\begin{itemize}
\item With probability $\rho_r(c)/Z$ it outputs $c$;
\item With probability $\rho_r(c-1)/Z$ it outputs $c-1$;
\item With probability $Z_0/Z$ it chooses $x$ from the restriction of the continuous normal distribution $D_r$ to the interval $[c,\infty)$. Let $y$
be the smallest element in $\Z+c$ that is larger than $x$. With probability $\rho_r(y)/\rho_r(x)$ output $y$, and otherwise repeat;
\item With probability $Z_1/Z$ it chooses $x$ from the restriction of the continuous normal distribution $D_r$ to the interval $(-\infty,c-1]$. Let $y$
be the largest element in $\Z+c$ that is smaller than $x$. With probability $\rho_r(y)/\rho_r(x)$ output $y$, and otherwise repeat.
\end{itemize}

Consider now one iteration of the procedure. The probability of outputting $c$ is $\rho_r(c)/Z$, that of outputting $c-1$ 
is $\rho_r(c-1)/Z$, that of outputting $c+k$ for some $k \ge 1$ is
\[
\frac{Z_0}{Z} \cdot \frac{1}{Z_0} \int_{c+k-1}^{c+k} \rho_r(x) \cdot \frac{\rho_r(c+k)}{\rho_r(x)}  \mathrm{d} x
= \frac{\rho_r(c+k)}{Z},
\]
and similarly, that of outputting $c-1-k$ for some $k \ge 1$ is 
$\rho_r(c-1-k)/Z$. From this it follows immediately that conditioned on outputting something, the output
distribution has support on $\Z+c$ and probability mass function proportional to $\rho_r$, and is therefore the desired discrete
Gaussian distribution $D_{\Z+c,r}$. Moreover, the probability of outputting something is
\[
\frac{\rho_r(\Z+c)}{Z} = \frac{\rho_r(\Z+c)}{Z_0+Z_1+\rho_r(c)+\rho_r(c-1)} \ge \frac{\rho_r(\Z+c)}{\rho_r(\Z+c)+\rho_r(c)+\rho_r(c-1)}
\ge \frac{1}{2}.
\]
Therefore at each iteration the procedure has probability of at least $1/2$ to terminate. As a result, the probability that the 
number of iterations is greater than $t$ is at most $2^{-t}$, and in particular, the expected number of iterations
is at most $2$.

\subsection{The general case}

For completeness, we start by recalling the $\sampleD$ procedure described in~\cite{DBLP:conf/stoc/GentryPV08}.
This is a recursive procedure that gets as input a basis $\matB=(\vecb_1,\ldots,\vecb_n)$ of an $n$-dimensional lattice $\Lambda=\lat(\matB)$, a parameter $r>0$, 
and a vector $\vecc \in \R^n$, and outputs a vector in $\Lambda+\vecc$ whose distribution is close to that of 
$D_{\Lambda+\vecc,r}$. Let $\widetilde{\vecb_1},\ldots,\widetilde{\vecb_n}$
be the Gram-Schmidt orthogonalization of $\vecb_1,\ldots,\vecb_n$, and let $\overline{\vecb_1},\ldots,\overline{\vecb_n}$ be the
normalized Gram-Schmidt vectors, i.e., $\overline{\vecb_i} = \widetilde{\vecb_i}/\|\widetilde{\vecb_i}\|$. The procedure is the following. 

\begin{enumerate}
\item Let $\vecc_n \leftarrow \vecc$. For $i \leftarrow n,\ldots,1$, do:
\begin{enumerate}
	\item Choose $v_i$ from $D_{\|\widetilde{\vecb_i}\| \Z+\inner{\vecc_i,\overline{\vecb_i}}, r}$ using the exact one-dimensional sampler.
	\item Let $\vecc_{i-1} \leftarrow \vecc_i + (v_i - \inner{\vecc_i, \overline{\vecb_i}}) \cdot \vecb_i / \|\widetilde{\vecb_i}\| - v_i \overline{\vecb_i}$.
\end{enumerate}	
\item Output $\vecv := \sum_{i=1}^n v_i \overline{\vecb_i}$.
\end{enumerate}

It is easy to verify that the procedure always outputs vectors in the coset $\Lambda+\vecc$. Moreover, the probability of
outputting any $\vecv \in \Lambda+\vecc$ is 
\[
\prod_{i=1}^n \frac{\rho_r(v_i)}{\rho_r(\|\widetilde{\vecb_i}\| \Z + \inner{\vecc_i,\overline{\vecb_i}})} = 
\frac{\rho_r(\vecv)}{\prod_{i=1}^n \rho_r(\|\widetilde{\vecb_i}\| \Z + \inner{\vecc_i,\overline{\vecb_i}})},
\]
where $\vecc_i$ are the values computed in the procedure when it outputs $\vecv$.
Notice that by Lemma~\ref{lem:boundonsmoothing} and our assumption on $r$, 
we have that $r \ge \eta_{1/(n+1)}(\|\widetilde{\vecb_i}\| \Z)$ for all $i$.
Therefore, by Lemma~\ref{lem:smoothing}, we have that for all $c \in \R$,
\begin{equation*}%\label{eq:deviationofgaussmass}
\rho_r(\|\widetilde{\vecb_i}\| \Z + c) \in \bracks*{1-\frac{2}{n+2},1} \rho_r(\|\widetilde{\vecb_i}\| \Z).
\end{equation*}

In order to get an exact sample, we combine the above procedure with rejection sampling.
Namely, we apply $\sampleD$ to obtain some vector $\vecv$. We then output $\vecv$
with probability 
\begin{equation}\label{eq:rejsamplingprob}
\frac{\prod_{i=1}^n \rho_r(\|\widetilde{\vecb_i}\| \Z + \inner{\vecc_i,\overline{\vecb_i}})}
{\prod_{i=1}^n \rho_r(\|\widetilde{\vecb_i}\| \Z)} \in \left( \left(1-\frac{2}{n+2}\right)^n,1 \right] \subseteq (e^{-2},1],
\end{equation}
and otherwise repeat. This probability can be efficiently computed, as we will show below.
As a result, in any given iteration the probability of outputting the vector $\vecv \in \Lambda+\vecc$ is 
\[
\frac{\rho_r(\vecv)}{\prod_{i=1}^n \rho_r(\|\widetilde{\vecb_i}\| \Z )}.
\]
Since the denominator is independent of $\vecv$, we obtain that in any given iteration,
conditioned on outputting something, the output is distributed according to
the desired distribution $D_{\Lambda+\vecc,r}$, and therefore this is also the overall
output distribution of our sampler. Moreover, by~\eqref{eq:rejsamplingprob}, 
the probability of outputting something in any given iteration is at least
$e^{-2}$, and therefore, the probability that the number of iterations is
greater than $t$ is at most $(1-e^{-2})^t$, and in particular, the expected
number of iterations is at most $e^{2}$.

It remains to show how to efficiently compute the probability in~\eqref{eq:rejsamplingprob}. By scaling, it suffices to show 
how to compute 
\[
\rho_r(\Z+c) = \sum_{k \in \Z} \exp(-\pi (k+c)^2 / r^2)
\] 
for any $r>0$ and $c \in [0,1)$. If $r<1$, the sum decays very fast, and we can achieve any desired $t$ bits of 
accuracy in time $\poly(t)$, which agrees with our notion of efficiently computing a real number (following, 
e.g., the treatment in~\cite[Section 1.4]{LovaszBook}). For $r \ge 1$, we use the Poisson summation formula (see, e.g.,~\cite[Lemma 2.8]{DBLP:journals/siamcomp/MicciancioR07}) to write
\[
\rho_r(\Z+c) = r \cdot \sum_{k \in \Z} \exp(-\pi k^2 r^2 + 2 \pi i c k) = r \cdot \sum_{k \in \Z} \exp(-\pi k^2 r^2) \cos(2 \pi c k),
\] 
which again decays fast enough so we can compute it to within any desired $t$ bits of accuracy in time $\poly(t)$.

%%% Local Variables: 
%%% mode: latex
%%% TeX-master: "lwehard"
%%% End: 

\paragraph{Acknowledgments:} We thank Elette Boyle, Adam Klivans, Vadim Lyubashevsky, Sasha Sherstov, Vinod Vaikuntanathan and Gilles Villard 
for useful discussions.

\bibliographystyle{alphaabbrvprelim}
\bibliography{lattices,crypto,ibe}

%\appendix

%\input{reallwe}

\end{document}